\documentclass[aps,preprint]{revtex4}%
\usepackage{amsfonts}
\usepackage{amsmath}
\usepackage{amssymb}
\usepackage{subfigure}
\usepackage{graphicx}%
\newtheorem{theorem}{Theorem}

\newtheorem{proposition}[theorem]{Proposition}

\newenvironment{proof}[1][Proof]{\noindent\textbf{#1.} }{\ \rule{0.5em}{0.5em}}
\begin{document}
\preprint{CTP-SCU/2017002}
\title{WKB Approximation for a Deformed Schrodinger-like Equation and its
Applications to Quasinormal Modes of Black Holes and Quantum Cosmology}
\author{Bochen Lv}
\email{bochennn@yahoo.com}
\author{Peng Wang}
\email{pengw@scu.edu.cn}
\author{Haitang Yang}
\email{hyanga@scu.edu.cn}
\affiliation{Center for Theoretical Physics, College of Physical Science and Technology,
Sichuan University, Chengdu, 610064, China}

\begin{abstract}
In this paper, we use the WKB approximation method to approximately solve a
deformed Schrodinger-like differential equation: $\left[  -\hbar^{2}%
\partial_{\xi}^{2}g^{2}\left(  -i\hbar\alpha\partial_{\xi}\right)
-p^{2}\left(  \xi\right)  \right]  \psi\left(  \xi\right)  =0$, which are
frequently dealt with in various effective models of quantum gravity, where
the parameter $\alpha$ characterizes effects of quantum gravity. For an
arbitrary function $g\left(  x\right)  $ satisfying several properties
proposed in the paper, we find the WKB solutions, the WKB connection formulas
through a turning point, the deformed Bohr--Sommerfeld quantization rule, and
the deformed tunneling rate formula through a potential barrier. Several
examples of applying the WKB approximation to the deformed quantum mechanics
are investigated. In particular, we calculate the bound states of the
P\"{o}schl-Teller potential and estimate the effects of quantum gravity on the
quasinormal modes of a Schwarzschild black hole. Moreover, the area quantum of
the black hole is considered via Bohr's correspondence principle. Finally, the
WKB solutions of the deformed Wheeler--DeWitt equation for a closed Friedmann
universe with a scalar field are obtained, and the effects of quantum gravity
on the probability of sufficient inflation is discussed in the context of the
tunneling proposal.

\end{abstract}
\maketitle
\tableofcontents



\section{Introduction}

The WKB approximation, named after Wentzel, Kramers, and Brillouin, is a
method for obtaining an approximate solution to a one-dimensional
Schrodinger-like differential equation:%
\begin{equation}
\left[  -\hbar^{2}\partial_{\xi}^{2}-p^{2}\left(  \xi\right)  \right]
\psi\left(  \xi\right)  =0,
\end{equation}
where the real function $p^{2}\left(  \xi\right)  $ can be either positively
or negatively valued. The WKB approximation has a wide range of applications.
Its principal applications are in calculating bound-state energies and
tunneling rates through potential barriers.

On the other hand, the construction of a quantum theory for gravity has posed
one of the most challenging problems of the theoretical physics. Although
there are various proposals for quantum gravity, a comprehensive theory is not
available yet. Rather than considering a full quantum theory of gravity, we
can instead study effective theories of quantum gravity. In various effective
models of quantum gravity, one always deals with a deformed Schrodinger-like
equation:%
\begin{equation}
\left[  P^{2}\left(  -i\hbar\partial_{\xi}\right)  -p^{2}\left(  \xi\right)
\right]  \psi\left(  \xi\right)  =0, \label{eq:DeformedEq}%
\end{equation}
where $P\left(  x\right)  =xg\left(  \alpha x\right)  $. The properties of the
function $g\left(  x\right)  $ will be discussed in section \ref{Sec:WKB}.
Note that the parameter $\alpha$ characterizes effects of quantum gravity. For
example, the deformed Schrodinger-like equation $\left(  \ref{eq:DeformedEq}%
\right)  $ could appear in two effective models, namely the Generalized
Uncertainty Principle (GUP) and the modified dispersion relation (MDR). We
will briefly show that how it appears in these two models in section
\ref{Sec:Examples}.

The WKB approximation in deformed space and its applications have been
considered in effective models of quantum gravity. For example, in the
framework of GUP, the WKB wave functions were obtained in
\cite{IN-Fityo:2005lwa}. Moreover, the deformed Bohr--Sommerfeld quantization
rule and tunneling rate formula were used to calculate bound states of
Harmonic oscillators and Hydrogen atoms \cite{IN-Fityo:2005lwa}, $\alpha
$-decay \cite{IN-Blado:2015ada,IN-Guo:2016btf}, quantum cosmogenesis
\cite{IN-Blado:2015ada}, the volume of a phase cell \cite{IN-Tao:2012fp}, and
electron emissions \cite{IN-Guo:2016btf}, for some specific function $g\left(
x\right)  $. In the context of both GUP and MDR, the deformed Bohr--Sommerfeld
quantization rule was used to compute the number of quantum states to find the
entanglement entropy of black holes in the brick wall
model\ \cite{IN-Wang:2015bwa,IN-Eune:2010kx,IN-Wang:2015zpa}. In
\cite{IN-Tao:2012fp,IN-Guo:2016btf}, we found the WKB connection formulas and
proved the deformed Bohr--Sommerfeld quantization rule and tunneling rate
formula for the $g\left(  x\right)  =\sqrt{1+x^{2}}$ case. In this paper, we
will consider the case with an arbitrary function $g\left(  x\right)  $, for
which the WKB connection formulas, Bohr--Sommerfeld quantization rule and
tunneling rate formula are obtained.

The organization of this paper is as follows. In section \ref{Sec:WKB}, the
deformed Schrodinger-like differential equation $\left(  \ref{eq:DeformedEq}%
\right)  $ are first approximately solved by the WKB method. After the
asymptotic behavior of exact solutions of eqn. $\left(  \ref{eq:DeformedEq}%
\right)  $ around turning points are found, we obtain the WKB connection
formulas through a turning point by matching these two solutions in the
overlap regions. Accordingly, the Bohr--Sommerfeld quantization rule and
tunneling rate formula are also given. In section \ref{Sec:Examples}, the
formulas obtained in section \ref{Sec:WKB} are used to investigate several
examples, namely harmonic oscillators, the Schwinger effect, the
P\"{o}schl-Teller potential, and quantum cosmology. Section \ref{Sec:Con} is
devoted to our discussion and conclusion. In the appendix, we plot the
contours used to compute the asymptotic behavior in the $g\left(  x\right)
=\frac{\tan\left(  x\right)  }{x}$ case.

\section{WKB Method}

\label{Sec:WKB}

We now apply the WKB method to approximately solve the deformed
Schrodinger-like differential equation $\left(  \ref{eq:DeformedEq}\right)  $.
In what follows, we choose that $\arg p\left(  \xi\right)  =0$ for
$p^{2}\left(  \xi\right)  >0$ and $\arg p\left(  \xi\right)  =\frac{\pi}{2}$
for $p^{2}\left(  \xi\right)  <0$. Moreover, we could rewrite $P\left(
x\right)  $ in terms of a new function $g\left(  x\right)  $ as
\begin{equation}
P\left(  x\right)  =xg\left(  \alpha x\right)  \text{.}%
\end{equation}
To study the WKB solutions of eqn. $\left(  \ref{eq:DeformedEq}\right)  $ and
the connection formulas through a turning point, we shall impose the following
conditions on the function $g\left(  x\right)  $:

\begin{itemize}
\item In the complex plane, $g\left(  z\right)  $ is assumed to be analytic
except for possible poles. We assume that $g\left(  0\right)  =1$.

\item For a positive real number $a>0$, each of the equations%
\begin{equation}
sg\left(  -ias\right)  =e^{i\pi k/2}\text{, with }k=0\text{, }1\text{,
}2\text{, }3, \label{eq:lamda}%
\end{equation}
possess only one regular solution $\lambda_{k}\left(  a\right)  $, which is
regular as $a\rightarrow0$ and becomes%
\begin{equation}
\lambda_{k}\left(  0\right)  =e^{i\pi k/2}\text{,}%
\end{equation}
and the possible runaway solutions $\eta_{k}^{\left(  i\right)  }\left(
a\right)  $, which becomes%
\begin{equation}
\eta_{k}^{\left(  i\right)  }\left(  a\right)  \sim\frac{\eta_{k}^{\left(
i\right)  }}{a},
\end{equation}
when $a\ll1$. We also assume that for small enough value of $a$, there exists
a $c_{1}>0$ such that for all possible $i$ and $k$,
\begin{equation}
\frac{c_{1}}{a}<\left\vert \eta_{k}^{\left(  i\right)  }\left(  a\right)
\right\vert \text{.}%
\end{equation}
If there is no runaway solution, we simply set $c_{1}=\infty$.

\item Finally, we assume that there exists a $c_{2}>0$ such that%
\begin{equation}
\left\vert g^{2}\left(  -is\right)  -1\right\vert \leq\frac{1}{2}\text{ for
}\left\vert s\right\vert \leq c_{2}.
\end{equation}

\end{itemize}

For example, $g\left(  x\right)  =1\pm x^{2}$ satisfies the above conditions
with $\exists0<c_{1}<1$ and $0<c_{2}\leq1/\sqrt{2}$. The function $g\left(
x\right)  =\frac{\tan x}{x}$ also satisfies the above conditions with
$\exists0<c_{1}<\pi$ and $0<c_{2}\leq\arctan\left(  3/2\right)  $.

\subsection{WKB Solutions}

To find an approximate solution via the WKB method, we could make the change
of variable%
\begin{equation}
\psi\left(  \xi\right)  =e^{\frac{iS\left(  \xi\right)  }{\hbar}}%
\end{equation}
for some function $S\left(  \xi\right)  $, which can be expanded in power
series over $\hbar$%
\begin{equation}
S\left(  \xi\right)  =S_{0}\left(  \xi\right)  +\frac{\hbar}{i}S_{1}\left(
\xi\right)  +\cdots. \label{eq:S}%
\end{equation}
Plugging eqn. $\left(  \ref{eq:S}\right)  $ into eqn. $\left(
\ref{eq:DeformedEq}\right)  $ gives
\begin{align}
P^{2}\left(  S_{0}^{\prime}\left(  \xi\right)  \right)   &  =p^{2}\left(
\xi\right)  ,\nonumber\\
\left[  P^{2}\left(  S_{0}^{\prime}\left(  \xi\right)  \right)  \right]
^{\prime}\frac{\hbar}{i}S_{1}\left(  \xi\right)   &  =\frac{i\hbar}{2}\left[
P^{2}\left(  S_{0}^{\prime}\left(  \xi\right)  \right)  \right]
^{\prime\prime}S_{0}^{\prime\prime}\left(  \xi\right)  , \label{eq:sexpand}%
\end{align}
where the prime denotes derivative with respect to the argument of the
corresponding function. The first equation in eqn. $\left(  \ref{eq:sexpand}%
\right)  $ can be solved for $S_{0}^{\prime}\left(  \xi\right)  $. In
particular, when $p^{2}\left(  \xi\right)  >0$,
\begin{equation}
S_{0}^{\prime}\left(  \xi\right)  =-i\left\vert p\left(  \xi\right)
\right\vert \lambda_{k}\left(  \alpha\left\vert p\left(  \xi\right)
\right\vert \right)  \text{, with }k=1\text{ and }3\text{,} \label{eq:k=1and3}%
\end{equation}
and when $p^{2}\left(  \xi\right)  <0$,%
\begin{equation}
S_{0}^{\prime}\left(  \xi\right)  =-i\left\vert p\left(  \xi\right)
\right\vert \lambda_{k}\left(  \alpha\left\vert p\left(  \xi\right)
\right\vert \right)  \text{, with }k=0\text{ and }2\text{,} \label{eq:k=0and2}%
\end{equation}
where $\lambda_{k}\left(  a\right)  $ are regular solutions of eqn. $\left(
\ref{eq:lamda}\right)  $. It is noteworthy that there are other possible
solutions, namely%
\begin{equation}
S_{0}^{\prime}\left(  \xi\right)  =-i\left\vert p\left(  \xi\right)
\right\vert \eta_{k}\left(  \alpha\left\vert p\left(  \xi\right)  \right\vert
\right)  .
\end{equation}
These solutions are called "runaways" solutions since they do not exist in the
limit of $\alpha\rightarrow0$. In \cite{WKB-Simon:1990ic}, it was argued that
these "runaways" solutions were not physical and hence should be discarded. A
similar argument was also given in the framework of the GUP
\cite{WKB-Ching:2012fu}. Therefore, we will discard the "runaways" solutions
and keep only the solutions $\left(  \ref{eq:k=1and3}\right)  $ and $\left(
\ref{eq:k=0and2}\right)  $ in this paper. Solving the second equation in eqn.
$\left(  \ref{eq:sexpand}\right)  $ gives%
\begin{equation}
S_{1}\left(  x\right)  =-\frac{1}{2}\ln\left\vert \left[  P^{2}\left(
x\right)  \right]  ^{\prime}|_{x=S_{0}^{\prime}\left(  \xi\right)
}\right\vert .
\end{equation}
The expression for the WKB solutions are%
\begin{equation}
\psi_{WKB}\left(  \xi\right)  =C_{1}\psi_{WKB}^{1}\left(  \xi\right)
+C_{3}\psi_{WKB}^{3}\left(  \xi\right)  \text{ for }p^{2}\left(  \xi\right)
>0 \label{eq:wkbp>0}%
\end{equation}
and
\begin{equation}
\psi_{WKB}\left(  \xi\right)  =C_{0}\psi_{WKB}^{0}\left(  \xi\right)
+C_{2}\psi_{WKB}^{2}\left(  \xi\right)  \text{ for }p^{2}\left(  \xi\right)
<0, \label{eq:wkbp<0}%
\end{equation}
where $C_{i}$ are constants, and we define%
\begin{equation}
\psi_{WKB}^{k}\left(  \xi\right)  =\frac{1}{\sqrt{\left\vert \left[
x^{2}g^{2}\left(  \alpha x\right)  \right]  ^{\prime}|_{x=-i\left\vert
p\left(  \xi\right)  \right\vert \lambda_{k}\left(  \alpha\left\vert p\left(
\xi\right)  \right\vert \right)  }\right\vert }}\exp\left(  \frac{1}{\hbar
}\int\left\vert p\left(  \xi\right)  \right\vert \lambda_{k}\left(
\alpha\left\vert p\left(  \xi\right)  \right\vert \right)  d\xi\right)  .
\end{equation}
These WKB solutions are valid if the RHS of the second equation in eqn.
$\left(  \ref{eq:sexpand}\right)  $ is much less than that of the first one.
Specifically, they are valid when%
\begin{equation}
\left\vert p^{2}\left(  \xi\right)  \right\vert \gg\frac{\hbar}{2}\left\vert
\left[  P^{2}\left(  S_{0}^{\prime}\left(  \xi\right)  \right)  \right]
^{\prime\prime}S_{0}^{\prime\prime}\left(  \xi\right)  \right\vert .
\label{eq:Condition}%
\end{equation}
However, the condition $\left(  \ref{eq:Condition}\right)  $ fails near a
turning point where $P\left(  x\right)  =0$. In the following of this section,
we will derived WKB connection formulas through the turning points.

\subsection{Connection Formulas}

We first investigate the asymptotic behavior of solutions of the differential
equation%
\begin{equation}
\partial_{\rho}^{2}g^{2}\left(  -i\tilde{\alpha}\partial_{\rho}\right)
\psi-\rho\psi=0, \label{eq:linearDE}%
\end{equation}
where $\tilde{\alpha}>0$. To solve this equation, it is useful to Laplace
transform it via%
\begin{equation}
\psi\left(  \rho\right)  =\int\nolimits_{C}e^{\rho t}\tilde{\psi}\left(
t\right)  dt, \label{eq:laplaceT}%
\end{equation}
where the contour $C$ in the complex plane will be discussed below. The
equation for $\tilde{\psi}\left(  t\right)  $ in terms of the complex variable
$t$ reads%
\begin{equation}
\frac{d\tilde{\psi}\left(  t\right)  }{dt}+t^{2}g^{2}\left(  -i\tilde{\alpha
}t\right)  \tilde{\psi}\left(  t\right)  =0, \label{eq:phitilde}%
\end{equation}
where we use the integration by parts to obtain the second term. Note the
integration by parts used in eqn. $\left(  \ref{eq:phitilde}\right)  $
requires that $e^{\rho t}\tilde{\psi}\left(  t\right)  $ vanishes at endpoints
of $C$. Up to an irrelevant pre-factor, its solution is%
\begin{equation}
\tilde{\psi}\left(  t\right)  =\exp\left(  -\int_{0}^{t}t^{\prime2}%
g^{2}\left(  -i\tilde{\alpha}t^{\prime}\right)  dt^{\prime}\right)  .
\end{equation}
To apply the saddle point method, we make the change of variables
$t=\left\vert \rho\right\vert ^{\frac{1}{2}}s$ and rewrite the Laplace
transformation in eqn. $\left(  \ref{eq:laplaceT}\right)  $ as%
\begin{equation}
\psi\left(  \rho\right)  =\left\vert \rho\right\vert ^{\frac{1}{2}}\int
_{C}\exp\left[  \left\vert \rho\right\vert ^{\frac{3}{2}}f_{\pm}\left(
s\right)  \right]  ds, \label{eq:interalsolution}%
\end{equation}
where we define $a=\tilde{\alpha}\left\vert \rho\right\vert ^{1/2}$ and%
\begin{equation}
f_{\pm}\left(  s\right)  =\pm s-\int_{0}^{s}s^{\prime2}g^{2}\left(
-ias^{\prime}\right)  ds^{\prime},
\end{equation}
with $+$ for $\rho>0$ and $-$ for $\rho<0$. The contour $C$ in eqn. $\left(
\ref{eq:interalsolution}\right)  $ is chosen so that the integrand vanishes at
endpoints of $C$.

Now consider a large circle $C_{R}$ of radius $R=\frac{c}{a}$, where
$c=\min\left\{  c_{1}\text{, }c_{2}\right\}  $. The saddle points of
$f_{+}\left(  s\right)  $ $\left(  f_{-}\left(  s\right)  \right)  $ are
$\lambda_{k}\left(  a\right)  $ and $\eta_{k}^{\left(  i\right)  }\left(
a\right)  $ with $k=0$ and $2$ $\left(  1\text{ and }3\right)  $. Thus, all
the saddle points except $\lambda_{k}\left(  a\right)  $ are outside the
circle $C_{R}$. To discuss the properties of the steepest descent contours
passing through $\lambda_{k}\left(  a\right)  $, we first prove two
propositions. In the following, let $C_{\lambda_{k}}$ denote the steepest
descent contours passing through $\lambda_{k}\left(  a\right)  $.

\begin{proposition}
For small $a$, if $C_{\lambda_{k}}$ intersects $C_{R}$ at $Re^{i\theta_{\ast}%
}$, then there exists an $n\in\left\{  0,1,2\right\}  $ such that
\begin{equation}
\left\vert \theta^{\ast}-\frac{2n\pi}{3}\right\vert \leq\frac{\pi}%
{18}+\mathcal{O}\left(  a\right)  .
\end{equation}
Moreover, one finds
\begin{equation}
\operatorname{Re}f_{\pm}\left(  Re^{i\theta_{\ast}}\right)  \lesssim
-\frac{c^{3}}{6a^{3}}+\mathcal{O}\left(  a^{-2}\right)  .
\end{equation}
\label{Pop1}
\end{proposition}

\begin{proof}
For $f_{\pm}\left(  s\right)  $, we have%
\begin{equation}
f_{\pm}\left(  s\right)  =\pm s-\frac{\left\vert s\right\vert ^{3}\rho\left(
s\right)  e^{i\left[  3\theta+\sigma\left(  s\right)  \right]  }}{3},
\end{equation}
where
\begin{align}
s  &  =\left\vert s\right\vert e^{i\theta},\nonumber\\
f\left(  s\right)  e^{i\alpha\left(  s\right)  }  &  =\frac{3}{a^{3}s^{3}}%
\int_{0}^{as}x^{2}\left[  g^{2}\left(  -ix\right)  -1\right]  dx,\\
\rho\left(  s\right)  e^{i\sigma\left(  s\right)  }  &  =1+f\left(  s\right)
e^{i\alpha\left(  s\right)  }.\nonumber
\end{align}
Since $\left\vert g^{2}\left(  -ix\right)  -1\right\vert \leq\frac{1}{2}$ for
$\left\vert x\right\vert \leq aR$, one finds for $\left\vert s\right\vert \leq
R$ that
\begin{equation}
f\left(  s\right)  \leq\frac{3}{\left\vert as\right\vert ^{3}}\int
_{0}^{\left\vert as\right\vert }\left\vert x^{2}\right\vert \left\vert
g^{2}\left(  -ix\right)  -1\right\vert d\left\vert x\right\vert \leq\frac
{1}{2},
\end{equation}
and hence%
\[
\frac{1}{2}\leq1-f\left(  s\right)  \leq\rho\left(  s\right)  \leq1+f\left(
s\right)  \leq\frac{3}{2},
\]%
\begin{equation}
\left\vert \sin\sigma\left(  s\right)  \right\vert \leq f\left(  s\right)
\Rightarrow\left\vert \sigma\left(  s\right)  \right\vert \leq\arcsin f\left(
s\right)  \leq\arcsin\frac{1}{2}=\frac{\pi}{6}.
\end{equation}

Suppose that the contour $C_{\lambda_{k}}$ intersects $C_{R}$ at $s_{\ast
}=Re^{i\theta_{\ast}}$. Since $C_{\lambda_{k}}$ is also a constant-phase
contour, $C_{\lambda_{k}}$ is determined by%
\begin{equation}
\operatorname{Im}f_{\pm}\left(  s\right)  =\operatorname{Im}f_{\pm}\left(
\lambda_{k}\left(  a\right)  \right)  .
\end{equation}
At $s=s_{\ast}$, this equation becomes%
\begin{equation}
\pm ca^{2}\sin\theta^{\ast}-c^{3}\frac{\rho\left(  s\right)  \sin\left[
3\theta^{\ast}+\sigma\left(  s^{\ast}\right)  \right]  }{3}=a^{3}%
\operatorname{Im}f_{\pm}\left(  \lambda_{k}\left(  a\right)  \right)  ,
\end{equation}
where we use $R=\frac{c}{a}$. For small $a$, one has $\operatorname{Im}f_{\pm
}\left(  \lambda_{k}\left(  a\right)  \right)  \sim\mathcal{O}\left(
a\right)  $ and hence%
\begin{equation}
\theta^{\ast}+\frac{\sigma\left(  s^{\ast}\right)  }{3}=\frac{2n\pi}%
{3}+\mathcal{O}\left(  a\right)  \text{ or }\frac{\left(  2n+1\right)  \pi}%
{3}+\mathcal{O}\left(  a\right)  \text{,} \label{eq:theta}%
\end{equation}
where $n\in\left\{  0,1,2\right\}  $. However for $\theta^{\ast}+\frac
{\sigma\left(  s^{\ast}\right)  }{3}=\frac{\left(  2n+1\right)  \pi}%
{3}+\mathcal{O}\left(  a\right)  $, we find at $s=s_{\ast}$ that%
\begin{equation}
\operatorname{Re}f_{\pm}\left(  s_{\ast}\right)  \sim\frac{\rho\left(
s_{\ast}\right)  c^{3}}{3a^{3}}\gg\operatorname{Re}f_{\pm}\left(  \lambda
_{k}\left(  a\right)  \right)  \sim\mathcal{O}\left(  a\right)  ,
\end{equation}
which contradicts $C_{\lambda_{k}}$ being the steepest descent contour. Thus,
for the steepest descent contour $C_{\lambda_{k}}$, eqn. $\left(
\ref{eq:theta}\right)  $ gives for some $n\in\left\{  0,1,2\right\}  $ that%
\[
\left\vert \theta^{\ast}-\frac{2n\pi}{3}\right\vert =\frac{\left\vert
\sigma\left(  s^{\ast}\right)  \right\vert }{3}+\mathcal{O}\left(  a\right)
\leq\frac{\pi}{18}+\mathcal{O}\left(  a\right)  .
\]
It can be easily shown that
\begin{equation}
\operatorname{Re}f_{\pm}\left(  s_{\ast}\right)  =\frac{\rho\left(  s_{\ast
}\right)  c^{3}}{3a^{3}}+\mathcal{O}\left(  a^{-2}\right)  \leq-\frac{c^{3}%
}{6a^{3}}+\mathcal{O}\left(  a^{-2}\right)  ,
\end{equation}
where we use $\rho\left(  s_{\ast}\right)  \geq\frac{1}{2}$.
\end{proof}

\begin{proposition}
On the circle $C_{R}$, $\operatorname{Re}f_{\pm}\left(  Re^{i\theta}\right)
\leq-\frac{c^{3}}{12a^{3}}+\mathcal{O}\left(  a^{-2}\right)  $ for $\left\vert
\theta-\frac{2n\pi}{3}\right\vert \leq\frac{\pi}{18}+\mathcal{O}\left(
a\right)  $, where $n\in\left\{  0,1,2\right\}  $. \label{Pop2}
\end{proposition}

\begin{proof}
Since $\left\vert \sigma\left(  s\right)  \right\vert \leq\frac{\pi}{6}$ on
the $C_{R}$, if $\left\vert \theta-\frac{2n\pi}{3}\right\vert \leq\frac{\pi
}{18}+\mathcal{O}\left(  a\right)  $, we have%
\begin{equation}
\left\vert \theta+\frac{\sigma\left(  s\right)  }{3}-\frac{2n\pi}%
{3}\right\vert \leq\frac{\pi}{9}+\mathcal{O}\left(  a\right)  ,
\end{equation}
which leads to%
\begin{equation}
\cos\left[  3\theta+\sigma\left(  s\right)  \right]  \geq\cos\left(  \frac
{\pi}{3}\right)  =\frac{1}{2}+\mathcal{O}\left(  a\right)  .
\end{equation}
Thus, it shows that on $\left\vert \theta-\frac{2n\pi}{3}\right\vert \leq
\frac{\pi}{18}+\mathcal{O}\left(  a\right)  $,
\begin{equation}
\operatorname{Re}f_{\pm}\left(  Re^{i\theta}\right)  =-\frac{R^{3}\rho\left(
Re^{i\theta}\right)  \cos\left[  3\theta+\sigma\left(  Re^{i\theta}\right)
\right]  }{3}+\mathcal{O}\left(  a^{-1}\right)  \leq-\frac{c^{3}}{12a^{3}%
}+\mathcal{O}\left(  a^{-2}\right)  ,
\end{equation}
where we use $\rho\left(  Re^{i\theta}\right)  \geq\frac{1}{2}$.
\end{proof}

\begin{figure}[tb]
\begin{centering}
\includegraphics[scale=0.5]{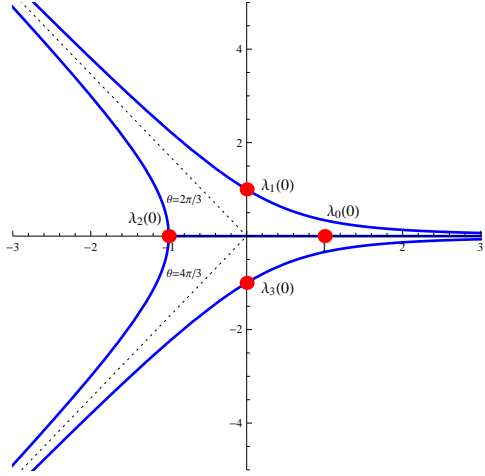}
\par\end{centering}
\caption{The saddle points (red dots) and steepest descent contours (blue
thick lines) of $f_{\pm}\left(  s\right)  $ when $a=0$.}%
\label{fig:a=0}%
\end{figure}

For $a=0$, we plot saddle points (red dots) of $f_{\pm}\left(  s\right)  $ and
the steepest descent contours (blue thick lines) passing through them in FIG.
$\ref{fig:a=0}$. When $a>0$, more possible saddle points and poles could
appear and these steepest descent contours could change dramatically around
them, e.g. a contour that goes to infinity in the case $a=0$ could change to
the one that ends at a new saddle point or a pole. However within $C_{R}$,
there are no new saddle points or poles, and hence $C_{\lambda_{k}}$ would
change continuously as $a$ is varied away from $0$. FIG. $\ref{fig:a=0}$ shows
that for $a=0$, $C_{\lambda_{2}}$ approaches $\theta=\frac{2\pi}{3}$ and
$\frac{4\pi}{3}$ for a large value of $\left\vert s\right\vert $. So when
$a>0$, $C_{\lambda_{2}}$ will intersect $C_{R}$ twice, and the intersections
$Re^{i\theta_{\ast}}$ are within $\left\vert \theta^{\ast}-\frac{2\pi}%
{3}\right\vert \leq\frac{\pi}{18}+\mathcal{O}\left(  a\right)  $ and
$\left\vert \theta^{\ast}-\frac{4\pi}{3}\right\vert \leq\frac{\pi}%
{18}+\mathcal{O}\left(  a\right)  $, respectively. Similarly, $C_{\lambda_{1}%
}$ intersects $C_{R}$ within $\left\vert \theta^{\ast}\right\vert \leq
\frac{\pi}{18}+\mathcal{O}\left(  a\right)  $ and $\left\vert \theta^{\ast
}-\frac{2\pi}{3}\right\vert \leq\frac{\pi}{18}+\mathcal{O}\left(  a\right)  $,
$C_{\lambda_{3}}$ intersects $C_{R}$ within $\left\vert \theta^{\ast
}\right\vert \leq\frac{\pi}{18}+\mathcal{O}\left(  a\right)  $ and $\left\vert
\theta^{\ast}-\frac{4\pi}{3}\right\vert \leq\frac{\pi}{18}+\mathcal{O}\left(
a\right)  $, and $C_{\lambda_{0}}$ ends at $\lambda_{2}\left(  a\right)  $ and
intersects $C_{R}$ within $\left\vert \theta^{\ast}\right\vert \leq\frac{\pi
}{18}+\mathcal{O}\left(  a\right)  $.

\begin{figure}[tb]
\begin{center}
\subfigure[{~Contours of $\psi_{1}\left(\rho\right)$. The contour used in the $\rho>0$ $\left(  \rho<0\right)  $ case is the one passing through the saddle point(s) $\lambda_{2}\left(  a\right)  $ $\left(\lambda_{1}\left(  a\right)  \text{ and }\lambda_{3}\left( a\right)  \right)$.  }]
{\includegraphics[width=0.45\textwidth]{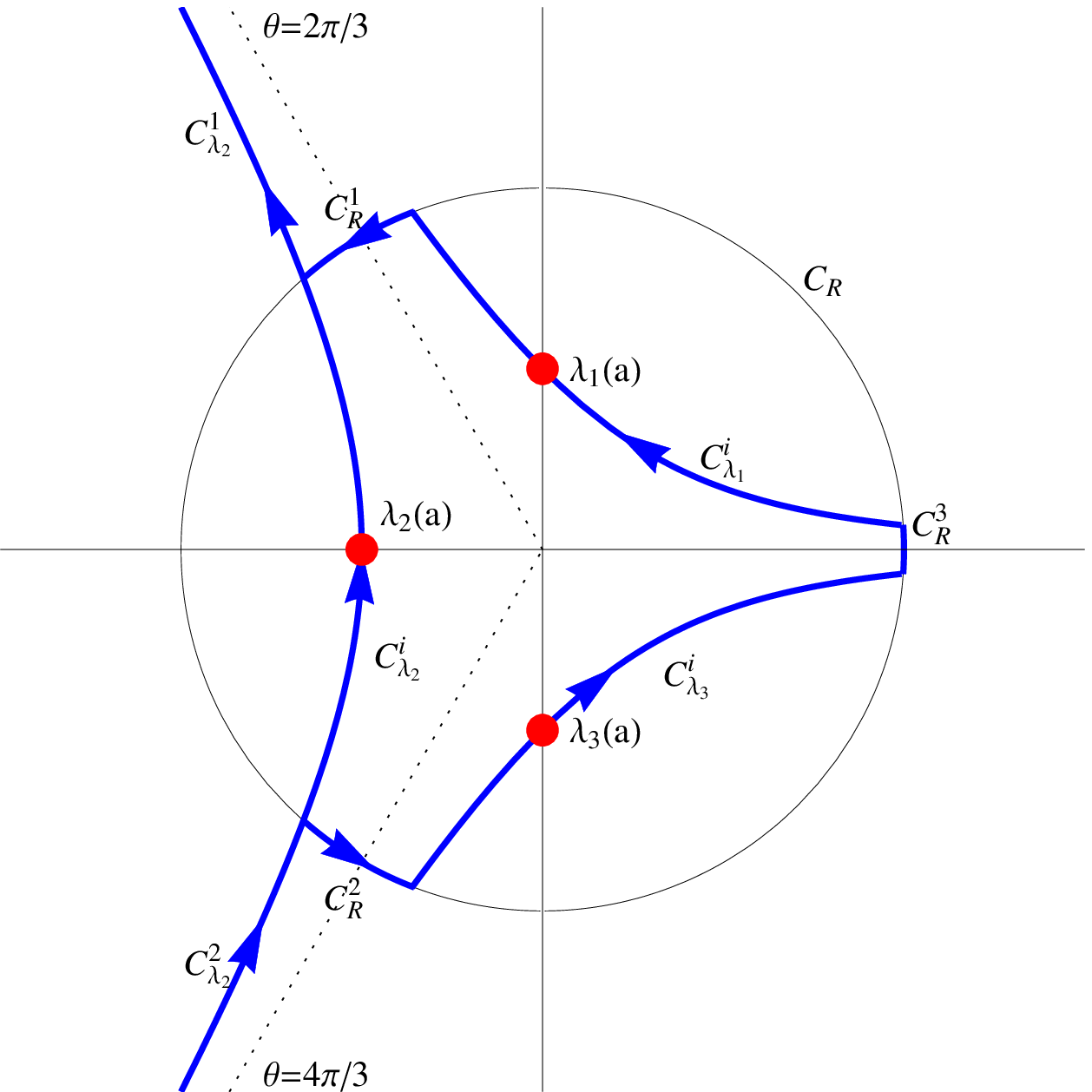}\label{fig:Generala:a}}
\subfigure[{~Contours of $\psi_{2}\left(  \rho\right)$. The contour used in the $\rho>0$ $\left(  \rho<0\right)  $ case is the one
passing through the saddle point(s) $\lambda_{0}\left(  a\right)  $ and
$\lambda_{2}\left(  a\right)  $ $\left(  \lambda_{3}\left(  a\right)  \right)
$. }]
{\includegraphics[width=0.45\textwidth]{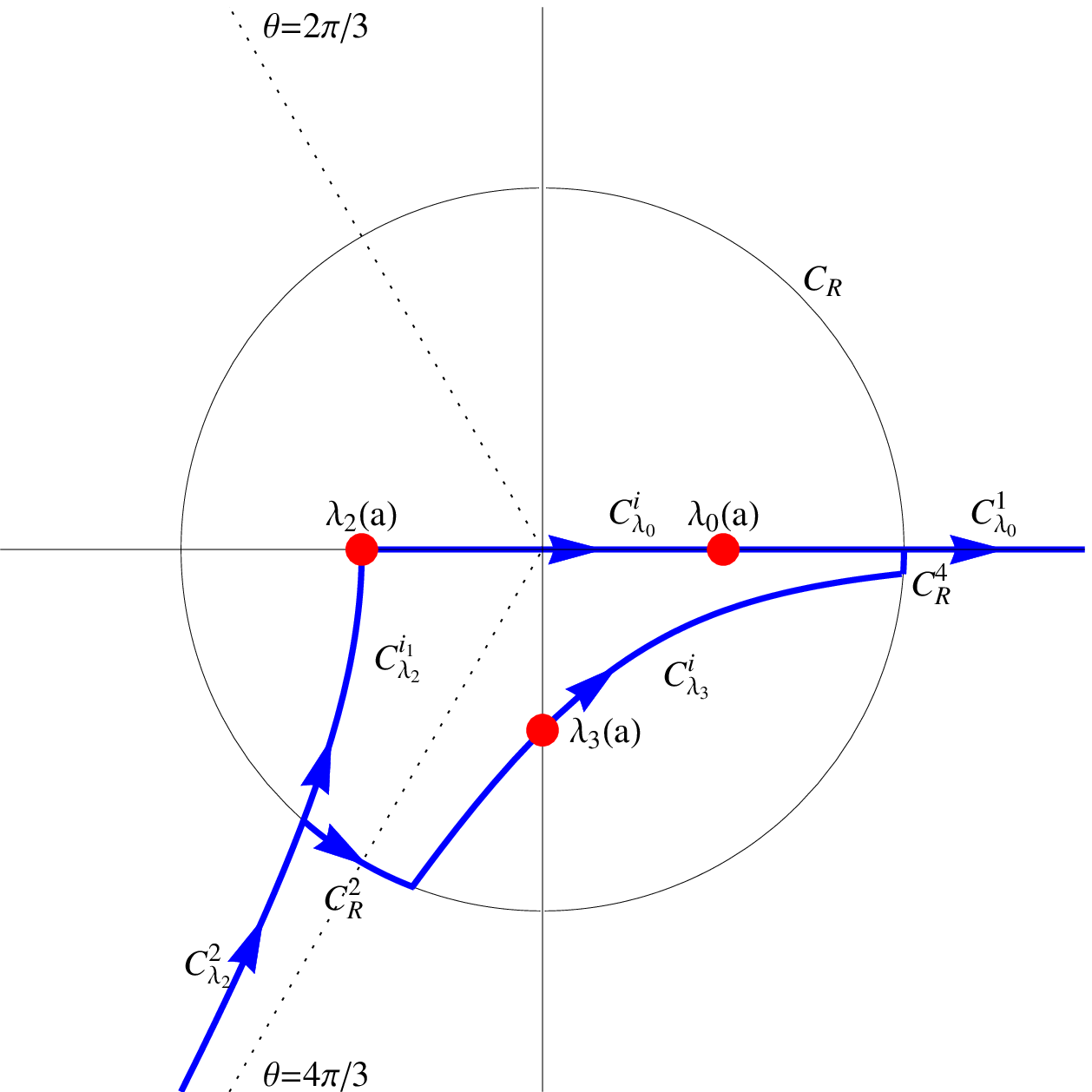}\label{fig:Generala:b}}
\end{center}
\caption{Contours (blue thick lines) and saddle points (red dots) of $\psi
_{1}\left(  \rho\right)  $ and $\psi_{2}\left(  \rho\right)  $ in the $\rho>0$
and $\rho<0$ cases.}%
\label{fig:Generala}%
\end{figure}

Since the saddle points and hence the solutions in eqn. $\left(
\ref{eq:interalsolution}\right)  $ depend on the sign of $\rho$, it is
convenient to choose different contours in the complex plane for $\rho>0$ and
$\rho<0$, making sure that they are deformable to each other. In FIG.
$\ref{fig:Generala:a}$, the contour considered in the $\rho>0$ case is the
steepest descent contour $C_{\lambda_{2}}$ through $\lambda_{2}\left(
a\right)  $. For simplicity, we assume $\arg\lambda_{k}\left(  a\right)
=k\pi/2$ to illustrate the contours in FIG. $\ref{fig:Generala}$. In FIG.
$\ref{fig:Generala}$, $C_{\lambda_{k}}^{i}$ denotes the part of $C_{\lambda
_{k}}$ inside the circle $C_{R}$ while $C_{\lambda_{k}}^{1,2}$ denotes the
parts of $C_{\lambda_{k}}$ outside $C_{R}$. Note that $f_{+}\left(  s\right)
\rightarrow-\infty$ when one moves away from the saddle point, and hence the
corresponding integrand in eqn. $\left(  \ref{eq:interalsolution}\right)  $
vanishes at endpoints of $C_{\lambda_{2}}$. Since $C_{\lambda_{2}}$ is a
steepest descent contour, when $1\ll\rho\ll\alpha^{-2}$, the dominant
contribution to the integral over $C_{\lambda_{2}}$ in eqn. $\left(
\ref{eq:interalsolution}\right)  $ comes from the neighborhood of the saddle
point $\lambda_{2}\left(  a\right)  $. Thus by the method of steepest descent,
one has for $1\ll\rho\ll\tilde{\alpha}^{-2}$ that
\begin{equation}
\psi_{1}\left(  \rho\right)  =\left\vert \rho\right\vert ^{\frac{1}{2}}%
\int_{C_{>}}\exp\left[  \left\vert \rho\right\vert ^{\frac{3}{2}}f_{+}\left(
s\right)  \right]  ds\sim I_{\lambda_{2}\left(  a\right)  }\text{,}
\label{eq:phi1plus}%
\end{equation}
where $C_{>}=C_{\lambda_{2}}$, and $I_{\lambda_{k}\left(  a\right)  }$ is the
contribution from the saddle point $\lambda_{k}\left(  a\right)  $. Using
Watson's lemma, we find
\begin{equation}
I_{\lambda_{k}\left(  a\right)  }\sim\frac{\sqrt{\pi}\exp\left[  \left\vert
\rho\right\vert ^{\frac{3}{2}}f_{\pm}\left(  \lambda_{k}\left(  a\right)
\right)  \right]  }{\left\vert \rho\right\vert ^{\frac{1}{4}}}\sqrt{\frac
{2}{\left[  s^{2}g^{2}\left(  -ias\right)  \right]  ^{\prime}|_{s=\lambda
_{k}\left(  a\right)  }}},
\end{equation}
where $k=0$ and $2$ is for $+$, and $k=1$ and $3$ for $-$. To study the
asymptotic behavior of $\psi_{1}\left(  \rho\right)  $ when $-1\gg\rho
\gg-\tilde{\alpha}^{-2}$, we consider
\begin{equation}
\psi_{1}\left(  \rho\right)  =\left\vert \rho\right\vert ^{\frac{1}{2}}%
\int_{C_{<}}\exp\left[  \left\vert \rho\right\vert ^{\frac{3}{2}}f_{-}\left(
s\right)  \right]  ds\text{,} \label{eq:phi1minus}%
\end{equation}
where the contour $C_{<}$ consists of $C_{\lambda_{2}}^{1}$, $C_{\lambda_{2}%
}^{2}$, $C_{\lambda_{1}}^{i}$, $C_{\lambda_{3}}^{i}$, $C_{R}^{1}$, $C_{R}^{2}%
$, and $C_{R}^{3}$, as shown in FIG. $\ref{fig:Generala:a}$. Since the
contributions from $C_{\lambda_{2}}^{1}$ and $C_{\lambda_{2}}^{2}$ are
neglected in eqn. $\left(  \ref{eq:phi1plus}\right)  $, they can also be
neglected in eqn. $\left(  \ref{eq:phi1minus}\right)  $. The contours
$C_{R}^{i}$ connect two adjacent steepest descent contour along the circle
$C_{R}$. Thus, propositions \ref{Pop1} and \ref{Pop2} implies that the
contributions from $C_{R}^{i}$ is
\begin{equation}
\left\vert \left\vert \rho\right\vert ^{\frac{1}{2}}\int_{C_{R}^{i}}%
\exp\left[  \left\vert \rho\right\vert ^{\frac{3}{2}}f_{-}\left(  s\right)
\right]  ds\right\vert \sim\frac{\pi c\left\vert \rho\right\vert ^{\frac{1}%
{2}}}{9a}\exp\left[  -\frac{c^{3}\left\vert \rho\right\vert ^{\frac{3}{2}}%
}{12a^{3}}\right]  ,
\end{equation}
where $i=1,2,3$. Since $\left\vert f_{-}\left(  s\right)  \right\vert
\sim\mathcal{O}\left(  1\right)  $ at $s=\lambda_{1}\left(  a\right)  $ and
$\lambda_{3}\left(  a\right)  $, the contributions from $C_{R}^{i}$ can also
be neglected. Thus\ considering the contributions from $C_{\lambda_{1}}^{i}$
and $C_{\lambda_{3}}^{i}$ around the saddle points $\lambda_{1}\left(
a\right)  $ and $\lambda_{3}\left(  a\right)  $, we find that when $-1\gg
\rho\gg-\tilde{\alpha}^{-2}$,
\begin{equation}
\psi_{1}\left(  \rho\right)  \sim I_{\lambda_{1}\left(  a\right)  }%
+I_{\lambda_{3}\left(  a\right)  }.
\end{equation}
Since there is no singularity inside $C_{R}$, the contour $C_{<}$ used in the
$\rho<0$ case can be deformed to $C_{>}$ in the $\rho>0$ case.

Similarly in FIG $\ref{fig:Generala:b}$ , we consider the contour
$C_{>}=C_{\lambda_{2}}^{2}+C_{\lambda_{2}}^{i_{1}}+C_{\lambda_{0}}%
^{i}+C_{\lambda_{0}}^{1}$ in the $\rho>0$ case and $C_{<}=C_{\lambda_{2}}%
^{2}+C_{R}^{2}+C_{\lambda_{3}}^{i}+C_{R}^{4}+C_{\lambda_{0}}^{1}$ in the
$\rho<0$ case. It is noteworthy that the contours $C_{>}$ and $C_{<}$ are
deformable to each other. As argued before, the contributions from $C_{R}^{2}$
and $C_{R}^{4}$ can be neglected. Since the leading contribution to the
integral over a steepest descent contour comes from a small vicinity of the
saddle point, the contributions from $C_{\lambda_{0}}^{1}$ and $C_{\lambda
_{2}}^{2}$ can also be neglected. Moreover, $\lambda_{2}\left(  a\right)  $ is
on the steepest descent contour $C_{\lambda_{0}}$ passing through $\lambda
_{0}\left(  a\right)  $, and hence $\left\vert I_{\lambda_{2}\left(  a\right)
}\right\vert \ll$ $\left\vert I_{\lambda_{0}\left(  a\right)  }\right\vert $.
So $I_{\lambda_{2}\left(  a\right)  }$ can be neglected for the integral over
$C_{>}$. Therefore when $1\ll\rho\ll\tilde{\alpha}^{-2}$, the asymptotic
behavior of the solution%
\begin{equation}
\psi_{2}\left(  \rho\right)  =\left\vert \rho\right\vert ^{\frac{1}{2}}%
\int_{C_{>}}\exp\left[  \left\vert \rho\right\vert ^{\frac{3}{2}}f_{-}\left(
s\right)  \right]  ds,
\end{equation}
is
\begin{equation}
\psi_{2}\left(  \rho\right)  \sim I_{\lambda_{0}\left(  a\right)  }\text{.}%
\end{equation}
When $-1\gg\rho\gg-\tilde{\alpha}^{-2}$, the asymptotic behavior of $\psi
_{2}\left(  \rho\right)  $ is
\begin{equation}
\psi_{2}\left(  \rho\right)  \sim I_{\lambda_{3}\left(  a\right)  }\text{.}%
\end{equation}
To better illustrate the contours, we plot these contours for $g\left(
x\right)  =\frac{\tan\left(  x\right)  }{x}$ in the appendix.

Now suppose that $p^{2}\left(  \xi\right)  $ has a simple (first order) at
$\xi=0$ and $F\equiv-\frac{dp^{2}\left(  \xi\right)  }{d\xi}|_{\xi=0}>0$. A
linear approximation to the potential $p^{2}\left(  \xi\right)  $\ near the
turning point $\xi=0$ is
\begin{equation}
p^{2}\left(  \xi\right)  \approx-F\xi. \label{eq:linearp}%
\end{equation}
In the vicinity of $\xi=0$, if we make change of variables $\xi=\ell_{F}\rho$
and $\alpha=\ell_{F}\tilde{\alpha}\hbar^{-1}$, where $\ell_{F}=\hbar
^{2/3}F^{-1/3}$, eqn. $\left(  \ref{eq:DeformedEq}\right)  $ becomes eqn.
$\left(  \ref{eq:linearDE}\right)  $. Thus in the region where eqn. $\left(
\ref{eq:linearp}\right)  $ holds, we conclude that $\psi_{1}\left(
\rho\right)  $ and $\psi_{2}\left(  \rho\right)  $ are solutions of eqn.
$\left(  \ref{eq:DeformedEq}\right)  $. On the other hand, we find for the
linear approximation of $p^{2}\left(  \xi\right)  $ that%
\begin{align}
&  \frac{1}{\hbar}\int_{0}^{\xi}\left\vert p\left(  \xi\right)  \right\vert
\lambda_{k}\left(  \alpha\left\vert p\left(  \xi\right)  \right\vert \right)
d\xi\nonumber\\
&  =-\frac{\ell_{F}^{3}}{\hbar^{3}}\int_{0}^{\frac{\hbar}{\ell_{F}}\sqrt
{-\rho}}\left\vert p\right\vert \lambda_{k}\left(  \alpha\left\vert
p\right\vert \right)  d\left(  p^{2}\right) \nonumber\\
&  =-\frac{\ell_{F}^{3}}{\hbar^{3}}\left[  -\frac{\hbar^{3}\rho\sqrt
{\left\vert \rho\right\vert }}{\ell_{F}^{3}}\lambda_{k}\left(  a\right)
+\int_{0}^{\frac{\hbar}{\ell_{F}}\sqrt{\left\vert \rho\right\vert }\lambda
_{k}\left(  a\right)  }u^{2}g^{2}\left(  -i\alpha u\right)  du\right]
\label{eq:wkbI}\\
&  =\left\vert \rho\right\vert ^{\frac{3}{2}}\left[  \text{sgn}\left(
\rho\right)  \lambda_{k}\left(  a\right)  -\int_{0}^{\lambda_{k}\left(
a\right)  }s^{2}g^{2}\left(  -ias\right)  ds\right] \nonumber\\
&  =\left\vert \rho\right\vert ^{\frac{3}{2}}f_{\pm}\left(  \lambda_{k}\left(
a\right)  \right)  ,\nonumber
\end{align}
where we use eqn. $\left(  \ref{eq:linearp}\right)  $ for $p^{2}$ in the
second line, $u=\left\vert p\right\vert \lambda_{k}\left(  \alpha\left\vert
p\right\vert \right)  $ in the third line, and $s=u/\frac{\hbar}{\ell_{F}%
}\sqrt{\left\vert \rho\right\vert }$ in the fourth line. Defining $\alpha
_{k}\left(  a\right)  $ and $\theta_{k}\left(  a\right)  $ as in%
\begin{equation}
\lambda_{k}\left(  a\right)  =\left\vert \lambda_{k}\left(  a\right)
\right\vert e^{i\left(  \pi k/2+\alpha_{k}\left(  a\right)  \right)  }\text{
and }\theta_{k}\left(  a\right)  =\arg\left[  1-ia\lambda_{k}^{2}\left(
a\right)  e^{-i\pi k/2}g^{\prime}\left(  -ia\lambda_{k}\left(  a\right)
\right)  \right]  ,
\end{equation}
one obtains for the linear approximation of $p^{2}\left(  \xi\right)  $ that
\begin{equation}
\left\vert \left[  x^{2}g^{2}\left(  \alpha x\right)  \right]  ^{\prime
}|_{x=-i\left\vert p\left(  \xi\right)  \right\vert \lambda_{k}\left(
\alpha\left\vert p\left(  \xi\right)  \right\vert \right)  }\right\vert
=\frac{\hbar}{\ell_{F}}\sqrt{\left\vert \rho\right\vert }\left[  s^{2}%
g^{2}\left(  -ias\right)  \right]  ^{\prime}|_{s=\lambda_{k}\left(  a\right)
}e^{-i\left(  \pi k/2+\theta_{k}\left(  a\right)  -\alpha_{k}\left(  a\right)
\right)  }, \label{eq:wkbF}%
\end{equation}
where $s=i\frac{\ell_{F}}{\hbar\sqrt{\left\vert \rho\right\vert }}x$. From
eqns. $\left(  \ref{eq:wkbI}\right)  $ and $\left(  \ref{eq:wkbF}\right)  $,
we have%
\begin{equation}
\frac{\exp\left(  \frac{1}{\hbar}\int\left\vert p\left(  \xi\right)
\right\vert \lambda_{k}\left(  \alpha\left\vert p\left(  \xi\right)
\right\vert \right)  d\xi\right)  }{\sqrt{\left\vert \left[  x^{2}g^{2}\left(
\alpha x\right)  \right]  ^{\prime}|_{x=-i\left\vert p\left(  \xi\right)
\right\vert \lambda_{k}\left(  \alpha\left\vert p\left(  \xi\right)
\right\vert \right)  }\right\vert }}\sim\sqrt{\frac{\ell_{F}}{2\pi\hbar}}%
C_{k}e^{i\left(  \pi k/4+\theta_{k}\left(  a\right)  /2-\alpha_{k}\left(
a\right)  /2\right)  }I_{\lambda_{k}\left(  a\right)  }. \label{eq:matching}%
\end{equation}

When $\left\vert \xi\right\vert \ll F^{-1}\alpha^{-2}$ $\left(  \alpha
\left\vert p\left(  \xi\right)  \right\vert \ll1\right)  $, the condition
$\left(  \ref{eq:Condition}\right)  $ for validity of the WKB approximation
becomes
\begin{equation}
\left\vert \xi\right\vert \gg\ell_{F}\text{.}%
\end{equation}
In terms of $\rho$, the WKB solutions $\left(  \ref{eq:wkbp>0}\right)  $ and
$\left(  \ref{eq:wkbp<0}\right)  $ for the linear approximation of
$p^{2}\left(  \xi\right)  $ are valid when $\tilde{\alpha}^{-2}\gg\left\vert
\rho\right\vert \gg$ $1$. However, we can approximate $\psi_{1}\left(
\rho\right)  $ and $\psi_{2}\left(  \rho\right)  $ by their leading asymptotic
behaviors for $\tilde{\alpha}^{-2}\gg\left\vert \rho\right\vert \gg$ $1$.
Since%
\begin{equation}
I_{\lambda_{1}\left(  a\right)  }+I_{\lambda_{3}\left(  a\right)  }%
\overset{-1\gg\rho\gg-\tilde{\alpha}^{-2}}{\sim}\psi_{1}\left(  \rho\right)
\overset{1\ll\rho\ll\alpha^{-2}}{\sim}I_{\lambda_{2}\left(  a\right)  },
\end{equation}
we use eqn. $\left(  \ref{eq:matching}\right)  $ to match WKB solutions
$\left(  \ref{eq:wkbp>0}\right)  $ and $\left(  \ref{eq:wkbp<0}\right)  $ with
$\psi_{1}\left(  \rho\right)  $ over the overlap region $\tilde{\alpha}%
^{-2}\gg\left\vert \rho\right\vert \gg$ $1$ and find that one connection
formula around the turning point is%
\begin{equation}
\psi_{WKB}\left(  \xi\right)  =\left\{
\begin{array}
[c]{c}%
Ce^{i\left[  \alpha_{2}\left(  a\right)  /2-\theta_{2}\left(  a\right)
/2\right]  }\psi_{WKB}^{2}\left(  \xi\right)  \text{ for }\xi>0\\
Ce^{i\left[  \pi/4+\alpha_{1}\left(  a\right)  /2-\theta_{1}\left(  a\right)
/2\right]  }\psi_{WKB}^{1}\left(  \xi\right)  +Ce^{i\left[  -\pi/4+\alpha
_{3}\left(  a\right)  /2-\theta_{3}\left(  a\right)  /2\right]  }\psi
_{WKB}^{3}\left(  \xi\right)  \text{ for }\xi<0
\end{array}
.\right.  \label{eq:wkbC1}%
\end{equation}
Similarly for $\psi_{2}\left(  \rho\right)  $, we find that another connection
formula is
\begin{equation}
\psi_{WKB}\left(  \xi\right)  =\left\{
\begin{array}
[c]{c}%
Ce^{i\left[  \alpha_{0}\left(  a\right)  /2-\theta_{0}\left(  a\right)
/2\right]  }\psi_{WKB}^{0}\left(  \xi\right)  \text{ for }\xi>0\\
Ce^{i\left[  -3\pi/4+\alpha_{3}\left(  a\right)  /2-\theta_{3}\left(
a\right)  /2\right]  }\psi_{WKB}^{3}\left(  \xi\right)  \text{ for }\xi<0
\end{array}
.\right.  \label{eq:wkbC2}%
\end{equation}

\subsection{Bohr-Sommerfeld Quantization and Tunneling Rates}

In the usual quantum mechanics, the Bohr-Sommerfeld quantization condition and
tunneling rates through potential barriers can be derived from the WKB
connection formulas. However, more conditions are needed to be imposed on the
function $g\left(  x\right)  $ to obtain the he Bohr-Sommerfeld quantization
condition and tunneling rates in our case. In fact, we further require that
$g\left(  z\right)  =g\left(  -z\right)  $ for $z\in C$, and both $g\left(
x\right)  $ and $\tilde{g}\left(  x\right)  \equiv g\left(  ix\right)  $ are
real functions when $x\in R$. Note that $g\left(  x\right)  =1\pm x^{2}$ and
$\tan x/x$ satisfy the above requirements. Under these requirements, the
solutions $\lambda_{k}\left(  a\right)  $ of eqn. $\left(  \ref{eq:lamda}%
\right)  $ satisfies the properties:

\begin{enumerate}
\item $\lambda_{0}\left(  a\right)  =-\lambda_{2}\left(  a\right)  $ and
$\lambda_{1}\left(  a\right)  =-\lambda_{3}\left(  a\right)  $,

\item For small enough value of $a$, one has that $\lambda_{k}\left(
a\right)  =\left\vert \lambda_{k}\left(  a\right)  \right\vert e^{ik\pi/2},$
\end{enumerate}

where the second property comes from the fact that for any real function
$f\left(  x\right)  $, the equation $xf\left(  ax\right)  =1$ always has a
real solution around $x=1$ if $a$ is small enough. Furthermore, these
properties imply that $\theta_{k}\left(  a\right)  =\alpha_{k}\left(
a\right)  =0$. In the region $\xi>0\,\ $where $p^{2}\left(  \xi\right)  <0$,
$\psi_{WKB}^{0}\left(  \xi\right)  $ is exponentially increasing away from the
turning point $\xi=0$ while $\psi_{WKB}^{2}\left(  \xi\right)  $ is
exponentially decreasing. In the region $\xi<0\,\ $where $p^{2}\left(
\xi\right)  >0$, $\psi_{WKB}^{1}\left(  \xi\right)  $ and $\psi_{WKB}%
^{3}\left(  \xi\right)  $ are oscillatory solutions and propagate toward and
away from the turning point, respectively.

Now suppose that $p^{2}\left(  \xi\right)  $ has two simple points at $\xi=A$
and $\xi=B$ with $A<B$. We also assume that $p^{2}\left(  \xi\right)  <0$ if
$\xi>B$ or $\xi<A$, and that $p^{2}\left(  \xi\right)  >0$ if $A<\xi<B$. To
study the boundary-value problem with $\psi\left(  \pm\infty\right)  =0$, we
consider the two-turning-point solutions by matching two one-turning-point
solutions: the first one is from $+\infty$ through $B$ and down to near $A$;
the second is $-\infty$ through $A$ and down to near $B$. We can use the WKB
connection formula $\left(  \ref{eq:wkbC1}\right)  $ to show that the first
one-turning-point solution that decays like%
\begin{equation}
C\psi_{WKB}^{2}\left(  \xi\right)  =\frac{C\exp\left(  \frac{1}{\hbar}\int
_{B}^{\xi}\left\vert p\left(  \xi^{\prime}\right)  \right\vert \lambda
_{2}\left(  \alpha\left\vert p\left(  \xi^{\prime}\right)  \right\vert
\right)  d\xi^{\prime}\right)  }{\sqrt{\left\vert \left[  x^{2}g^{2}\left(
\alpha x\right)  \right]  ^{\prime}|_{x=i\left\vert p\left(  \xi\right)
\right\vert \lambda_{2}\left(  \alpha\left\vert p\left(  \xi\right)
\right\vert \right)  }\right\vert }},
\end{equation}
as $\xi\rightarrow+\infty$ behaves like%
\begin{align}
&  \frac{2C}{\sqrt{\left\vert \left[  x^{2}g^{2}\left(  \alpha x\right)
\right]  ^{\prime}|_{x=\left\vert p\left(  \xi\right)  \lambda_{1}\left(
\alpha\left\vert p\left(  \xi\right)  \right\vert \right)  \right\vert
}\right\vert }}\sin\left(  \frac{1}{\hbar}\int_{\xi}^{B}\left\vert p\left(
\xi^{\prime}\right)  \lambda_{1}\left(  \alpha\left\vert p\left(  \xi^{\prime
}\right)  \right\vert \right)  \right\vert d\xi^{\prime}+\frac{\pi}{4}\right)
\nonumber\\
&  =-\frac{2C\sin\left[  \frac{1}{\hbar}\int_{A}^{\xi}\left\vert p\left(
\xi^{\prime}\right)  \lambda_{1}\left(  \alpha\left\vert p\left(  \xi^{\prime
}\right)  \right\vert \right)  \right\vert d\xi^{\prime}+\frac{\pi}%
{4}-\left\{  \frac{1}{\hbar}\int_{A}^{B}\left\vert p\left(  \xi^{\prime
}\right)  \lambda_{1}\left(  \alpha\left\vert p\left(  \xi^{\prime}\right)
\right\vert \right)  \right\vert d\xi^{\prime}+\frac{\pi}{2}\right\}  \right]
}{\sqrt{\left\vert \left[  x^{2}g^{2}\left(  \alpha x\right)  \right]
^{\prime}|_{x=\left\vert p\left(  \xi\right)  \lambda_{1}\left(
\alpha\left\vert p\left(  \xi\right)  \right\vert \right)  \right\vert
}\right\vert }}, \label{eq:wkbAB}%
\end{align}
in the region between $A$ and $B$. Similarly, the second one-turning-point
solution that decays like%
\begin{equation}
\frac{C^{\prime}\exp\left(  \frac{1}{\hbar}\int_{\xi}^{A}\left\vert p\left(
\xi^{\prime}\right)  \right\vert \lambda_{2}\left(  \alpha\left\vert p\left(
\xi^{\prime}\right)  \right\vert \right)  d\xi^{\prime}\right)  }%
{\sqrt{\left\vert \left[  x^{2}g^{2}\left(  \alpha x\right)  \right]
^{\prime}|_{x=i\left\vert p\left(  \xi\right)  \right\vert \lambda_{2}\left(
\alpha\left\vert p\left(  \xi\right)  \right\vert \right)  }\right\vert }},
\end{equation}
as $\xi\rightarrow-\infty$ behaves like%
\begin{equation}
\frac{2C^{\prime}}{\sqrt{\left\vert \left[  x^{2}g^{2}\left(  \alpha x\right)
\right]  ^{\prime}|_{x=\left\vert p\left(  \xi\right)  \lambda_{1}\left(
\alpha\left\vert p\left(  \xi\right)  \right\vert \right)  \right\vert
}\right\vert }}\sin\left(  \frac{1}{\hbar}\int_{A}^{\xi}\left\vert p\left(
\xi^{\prime}\right)  \lambda_{1}\left(  \alpha\left\vert p\left(  \xi^{\prime
}\right)  \right\vert \right)  \right\vert d\xi^{\prime}+\frac{\pi}{4}\right)
, \label{eq:wkbBA}%
\end{equation}
in the region between $A$ and $B$. In order that the two solutions in eqns.
$\left(  \ref{eq:wkbAB}\right)  $ and $\left(  \ref{eq:wkbBA}\right)  $ match
over the region between $A$ and $B$, we require that the expression in the
curly bracket of eqn. $\left(  \ref{eq:wkbAB}\right)  $ is an integral
multiple of $\pi$. Therefore, we derive the Bohr-Sommerfeld quantization
condition:%
\begin{equation}
\frac{1}{\hbar}\int_{A}^{B}\left\vert p\left(  \xi\right)  \lambda_{1}\left(
\alpha\left\vert p\left(  \xi\right)  \right\vert \right)  \right\vert
d\xi=\left(  n+\frac{1}{2}\right)  \pi+\mathcal{O}\left(  \hbar\right)  ,
\label{eq:BS}%
\end{equation}
where $n$ is a nonnegative integer.

We now consider WKB description of tunneling, in which $p\left(
-\infty\right)  =$ $p\left(  +\infty\right)  >0$ and $p^{2}\left(  \xi\right)
$ vanishes at two turning points $x=A\ $and $x=B$. Moreover, there are two
classical allowed regions $p^{2}\left(  \xi\right)  >0$, Region I with $\xi<A$
and Region III with $\xi>B$, and one forbidden region $p^{2}\left(
\xi\right)  <0$, Region II with $A<\xi<B$. To describe tunneling, we need to
choose appropriate boundary conditions in the classical allowed regions. We
postulate that there is only a transmitted wave in Region III:%
\begin{align}
F\psi_{WKB}^{3}\left(  \xi\right)   &  =\frac{F\exp\left(  \frac{1}{\hbar}%
\int_{0}^{B-\xi}\left\vert p\left(  B-\xi^{\prime}\right)  \right\vert
\lambda_{3}\left(  \alpha\left\vert p\left(  B-\xi^{\prime}\right)
\right\vert \right)  d\xi^{\prime}\right)  }{\sqrt{\left\vert \left[
x^{2}g^{2}\left(  \alpha x\right)  \right]  ^{\prime}|_{x=\left\vert p\left(
\xi\right)  \lambda_{3}\left(  \alpha\left\vert p\left(  \xi\right)
\right\vert \right)  \right\vert }\right\vert }}\nonumber\\
&  =\frac{F\exp\left(  \frac{i}{\hbar}\int_{B}^{\xi}\left\vert p\left(
\xi^{\prime}\right)  \lambda_{1}\left(  \alpha\left\vert p\left(  \xi^{\prime
}\right)  \right\vert \right)  \right\vert d\xi^{\prime}\right)  }%
{\sqrt{\left\vert \left[  x^{2}g^{2}\left(  \alpha x\right)  \right]
^{\prime}|_{x=\left\vert p\left(  \xi\right)  \lambda_{1}\left(
\alpha\left\vert p\left(  \xi\right)  \right\vert \right)  \right\vert
}\right\vert }}.
\end{align}
Using the WKB connection formula $\left(  \ref{eq:wkbC2}\right)  $, we find
that the WKB solution in Region II is
\begin{equation}
\frac{Fe^{3i\pi/4}\exp\left(  \frac{1}{\hbar}\int_{\xi}^{B}\left\vert p\left(
\xi^{\prime}\right)  \right\vert \lambda_{0}\left(  \alpha\left\vert p\left(
\xi^{\prime}\right)  \right\vert \right)  d\xi^{\prime}\right)  }%
{\sqrt{\left\vert \left[  x^{2}g^{2}\left(  \alpha x\right)  \right]
^{\prime}|_{x=i\left\vert p\left(  \xi\right)  \right\vert \lambda_{0}\left(
\alpha\left\vert p\left(  \xi\right)  \right\vert \right)  }\right\vert }%
}=\frac{Fe^{3i\pi/4}e^{\eta}\exp\left(  \frac{1}{\hbar}\int_{A}^{\xi
}\left\vert p\left(  \xi^{\prime}\right)  \right\vert \lambda_{2}\left(
\alpha\left\vert p\left(  \xi^{\prime}\right)  \right\vert \right)
d\xi^{\prime}\right)  }{\sqrt{\left\vert \left[  x^{2}g^{2}\left(  \alpha
x\right)  \right]  ^{\prime}|_{x=i\left\vert p\left(  \xi\right)  \right\vert
\lambda_{0}\left(  \alpha\left\vert p\left(  \xi\right)  \right\vert \right)
}\right\vert }},
\end{equation}
where
\begin{equation}
\eta=\frac{1}{\hbar}\int_{A}^{B}\left\vert p\left(  \xi\right)  \right\vert
\lambda_{0}\left(  \alpha\left\vert p\left(  \xi\right)  \right\vert \right)
d\xi.
\end{equation}
In Region I, the WKB approximation solution includes a wave incident the
barrier and a reflected wave:%
\begin{equation}
\frac{Fe^{\eta}\left[  e^{i\pi/2}\exp\left(  \frac{i}{\hbar}\int_{A}^{\xi
}\left\vert p\left(  \xi^{\prime}\right)  \lambda_{1}\left(  \alpha\left\vert
p\left(  \xi^{\prime}\right)  \right\vert \right)  \right\vert d\xi^{\prime
}\right)  +e^{i\pi}\exp\left(  -\frac{i}{\hbar}\int_{A}^{\xi}\left\vert
p\left(  \xi^{\prime}\right)  \lambda_{1}\left(  \alpha\left\vert p\left(
\xi^{\prime}\right)  \right\vert \right)  \right\vert d\xi^{\prime}\right)
\right]  }{\sqrt{\left\vert \left[  x^{2}g^{2}\left(  \alpha x\right)
\right]  ^{\prime}|_{x=\left\vert p\left(  \xi\right)  \lambda_{1}\left(
\alpha\left\vert p\left(  \xi\right)  \right\vert \right)  \right\vert
}\right\vert }},
\end{equation}
where the first term in the square bracket is the incident wave with the
amplitude $A=Fe^{\eta}e^{i\pi/2}$. Therefore, the transmission probability is%
\begin{equation}
T=\frac{\left\vert F\right\vert ^{2}}{\left\vert A\right\vert ^{2}}\sim
e^{-2\eta}. \label{eq:tunneling}%
\end{equation}

\section{Examples}

\label{Sec:Examples}

In this section, we use the results obtained in the previous section to
discuss some interesting examples in the deformed quantum mechanics. To
compare results in the literature, we first show that how a deformed
Schrodinger-like equation appears in the two effective models of quantum
gravity mentioned in Introduction. The first one is the GUP
\cite{IN-Maggiore:1993kv,IN-Kempf:1994su}, derived from the modified
fundamental commutation relation:
\begin{equation}
\lbrack X,P]=i\hbar f\left(  P\right)  ,\label{eq:1dGUP}%
\end{equation}
where $f\left(  P\right)  $ is some function. This model is inspired by the
prediction of the existence of a minimal length in various theories of quantum
gravity, such as string theory, loop quantum gravity and quantum geometry
\cite{IN-Townsend:1977xw,IN-Amati:1988tn,IN-Konishi:1989wk}. For example, if
$f\left(  P\right)  =1+\beta P^{2}$, the minimal measurable length is
\begin{equation}
l_{\min}=\hbar\sqrt{\beta}.
\end{equation}
The GUP has been extensively studied recently, see for example
\cite{IN-Chang:2001bm,IN-Brau:1999uv,IN-Das:2008kaa,IN-Hossenfelder:2003jz,IN-Ali:2009zq,IN-Li:2002xb,IN-Brau:2006ca,IN-Scardigli:1999jh,IN-Scardigli:2014qka}%
. For a review of the GUP, see \cite{IN-Chen:2014xgj,IN-Hossenfelder:2012jw}.
To study 1D quantum mechanics with the deformed commutators $(\ref{eq:1dGUP}%
)$, one can exploit the following representation for $X$ and $P$ in the the
position representation:%
\begin{equation}
X=X_{0}\text{ and }P=P\left(  \frac{\hbar}{i}\frac{\partial}{\partial
x}\right)  ,
\end{equation}
where the function $P\left(  x\right)  $ is the solution of the differential
equation $\frac{dP\left(  x\right)  }{dx}=f\left(  P\right)  $. Usually, we
could write $P\left(  x\right)  $ in terms of the function $g\left(  x\right)
$ as
\begin{equation}
P\left(  x\right)  =xg\left(  \alpha x\right)  ,
\end{equation}
where $\alpha$ is a parameter and can be related to the minimal length:
$l_{\min}=\hbar\alpha$. In the $f\left(  P\right)  =1+\beta P^{2}$ case, one
finds that
\begin{equation}
g\left(  x\right)  =\frac{\tan\left(  x\right)  }{x},
\end{equation}
with $\alpha=\sqrt{\beta}$.

The second is the MDR. It is believed that the trans-Planckian physics
manifests itself in certain modifications of the existing models. Thus, even
though a complete theory of quantum gravity is not yet available, we can use a
\textquotedblleft bottom-to-top approach\textquotedblright\ to probe the
possible effects of quantum gravity on our current theories and experiments
\cite{IN-AmelinoCamelia:2004hm}. One possible way of how such an approach
works is via Planck-scale modifications of the usual energy-momentum
dispersion relation
\begin{equation}
p^{2}=E^{2}-m^{2}, \label{eq:uMDR}%
\end{equation}
whose possibility has been considered in the quantum-gravity literature
\cite{IN-AmelinoCamelia:1997gz,IN-Garay:1998wk,IN-AmelinoCamelia:2002wr,IN-Magueijo:2002am}%
. The modified dispersion relation (MDR) has been reviewed in the framework of
Lorentz violating theories in \cite{IN-Mattingly:2005re,IN-Liberati:2013xla}.
It has also been shown that the MDR might play a role in astronomical and
cosmological observations, such as the threshold anomalies of ultra high
energy cosmic rays and TeV photons
\cite{IN-AmelinoCamelia:1997gz,IN-Colladay:1998fq,IN-Coleman:1998ti,IN-AmelinoCamelia:2000zs,IN-Jacobson:2001tu,IN-Jacobson:2003bn}%
. In most cases, the MDR could take the form%
\begin{equation}
E^{2}=p^{2}g^{2}\left(  \alpha p\right)  +m^{2}, \label{eq:MDR}%
\end{equation}
where $\alpha=\Lambda^{-1}$, and $\Lambda$ is the cut off scale which
characterizes the new physics in Planck scale. To obtain the deformed wave
equations, we can define the modified differential operator by%
\begin{equation}
P=P_{0}g\left(  \alpha P_{0}\right)  , \label{eq:pp0}%
\end{equation}
where $P_{0}=\frac{\hbar}{i}\frac{\partial}{\partial x}$, and replace $P_{0}$
with $P$ \cite{IN-Corley:1996ar,IN-Pavlopoulos:1967dm}. Another way to obtain
eqn. $\left(  \ref{eq:pp0}\right)  $ is using effective field theories (EFT).
In fact, one has to break or modify the global Lorentz symmetry in the
classical limit of the quantum gravity to have a MDR. There are several
possibilities for breaking or modifying the Lorentz symmetry, one of which is
that Lorentz invariance is spontaneously broken by extra tensor fields taking
on vacuum expectation values \cite{IN-Colladay:1998fq}. The most conservative
approach for a framework in which to describe MDR is considering an EFT, in
which modifications to the dispersion relation can be described by the higher
dimensional operators. In \cite{In-Wang:2015lqa}, we constructed a such EFT
for a scalar field, which only contained the kinetic terms and the usual
minimal gravitational couplings. It showed for the MDR $\left(  \ref{eq:MDR}%
\right)  $ that the deformed Klein-Gordon equation could be obtained via
making the replacement $\left(  \ref{eq:pp0}\right)  $. In this section, we
take Planck units $c=G=\hbar=k=1$.

\subsection{Harmonic Oscillator}

We first study a simple example, bound states of an harmonic oscillator in the
potential $V\left(  x\right)  =\frac{m\omega^{2}x^{2}}{2}$. For the harmonic
oscillator, the deformed Schrodinger equation is
\begin{equation}
\left[  -\partial_{x}^{2}g^{2}\left(  -i\alpha\partial_{x}\right)
-p^{2}\left(  x\right)  \right]  \psi\left(  x\right)  =0, \label{eq:HOeqn}%
\end{equation}
where $p^{2}\left(  x\right)  =m^{2}\omega^{2}\left(  x_{0}^{2}-x^{2}\right)
$ and $x_{0}=\sqrt{\frac{2E}{m\omega^{2}}}$. Considering $g\left(  x\right)
=\frac{\tan\left(  x\right)  }{x}$, one then has%
\begin{equation}
\lambda_{1}\left(  a\right)  =\frac{i}{a}\arctan a.
\end{equation}
In this case, the Bohr-Sommerfeld quantization condition becomes%
\begin{equation}
2\int_{0}^{x_{0}}\left\vert p\left(  x\right)  \frac{i}{\alpha\left\vert
p\left(  x\right)  \right\vert }\arctan\left(  \alpha\left\vert p\left(
x\right)  \right\vert \right)  \right\vert dx=\left(  n+\frac{1}{2}\right)
\pi,
\end{equation}
which gives the energy levels of bound states%
\begin{equation}
E_{n}=E_{0,n}\left(  1+\frac{\alpha^{2}mE_{0,n}}{2}\right)  , \label{eq:HOWKB}%
\end{equation}
for $n=0,1,2\cdots$, where $E_{0,n}=\left(  n+\frac{1}{2}\right)  \omega$.

In \cite{EX-Chang:2001kn}, the differential equation $\left(  \ref{eq:HOeqn}%
\right)  $ with $g\left(  x\right)  =\frac{\tan\left(  x\right)  }{x}$ was
solved exactly in the momentum space, and the exact energy levels were given
by%
\begin{align}
E_{n}  &  =\omega\left[  \left(  n+\frac{1}{2}\right)  \sqrt{1+\frac{\beta
^{2}m^{2}\omega^{2}}{4}}+\left(  n^{2}+n+\frac{1}{2}\right)  \frac{\beta
m\omega}{2}\right] \nonumber\\
&  =E_{0,n}\left[  \sqrt{1+\alpha^{4}m^{2}E_{0,n}^{2}\left(  \frac{\omega
}{2E_{0,n}}\right)  ^{2}}+\frac{\alpha^{2}mE_{0,n}}{2}+\frac{\alpha
^{2}mE_{0,n}}{2}\left(  \frac{\omega}{2E_{0,n}}\right)  ^{2}\right]  ,
\label{eq:HOexact}%
\end{align}
where $\beta$ in \cite{EX-Chang:2001kn} is our $\alpha^{2}$. The WKB
approximation is a good approximation when the de Broglie wavelength $\lambda$
of a particle is smaller than the characteristic length $L$ of the potential.
Thus, the higher order WKB corrections are suppressed by powers of
$\frac{\lambda}{L}$ relative to the leading term. For the harmonic oscillator
in the $n$th energy level, we find that%
\begin{equation}
\frac{\lambda}{L}\sim\frac{\frac{1}{\sqrt{2mE_{0,n}}}}{x_{0}}=\frac{\omega
}{2E_{0,n}}\text{.}%
\end{equation}
Therefore, the terms proportional to $\left(  \frac{\omega}{2E_{0,n}}\right)
^{2}$ in eqn. $\left(  \ref{eq:HOexact}\right)  $ are the higher order
corrections, and hence the WKB result $\left(  \ref{eq:HOWKB}\right)  $ agrees
with the leading term of the WKB expansion of the exact result $\left(
\ref{eq:HOexact}\right)  $. It is noteworthy that the momentum representation
of the position operator is quadratic in the $g\left(  x\right)  =\frac
{\tan\left(  x\right)  }{x}$ case, and hence eqn. $\left(  \ref{eq:HOeqn}%
\right)  $ can be solved exactly. However for a generic case, the WKB
approximation could provide a simple way to estimate the quantum gravity's corrections.

\subsection{Schwinger Effect}

The Schwinger effect \cite{EX-Schwinger:1951nm} is an example of creation of
particles by external fields, which consists in the creation of
electron--positron pairs by a strong electric field. The WKB approximation and
the Dirac sea picture can be used to illustrate its main physical features in
a heuristic way. Suppose the electrostatic potential is%
\begin{equation}
V\left(  x\right)  =\left\{
\begin{array}
[c]{l}%
0\text{ \ \ \ \ \ \ \ \ }x<0\\
-\mathcal{E}x\text{ \ \ \ }0<x<L\\
-\mathcal{E}L\text{\ \ \ }L<x
\end{array}
\right.  .
\end{equation}
If $e\mathcal{E}L>2m$, there exist states in the Dirac sea with $x>L$ having
the same energy as some positive energy states in the region $x<0$. The
electrons with energy $m\leq E\leq e\mathcal{E}L-m$ in the Dirac sea could
tunnel through the classically forbidden region leaving a hole behind, which
can be described as the production of an electron--positron pair out of the
vacuum by the effect of the electric field.

\begin{figure}[tb]
\begin{center}
\subfigure[{~The plot of $p^{2}\left(  x\right)  $ for the Schwinger effect.  }]
{\includegraphics[width=0.45\textwidth]{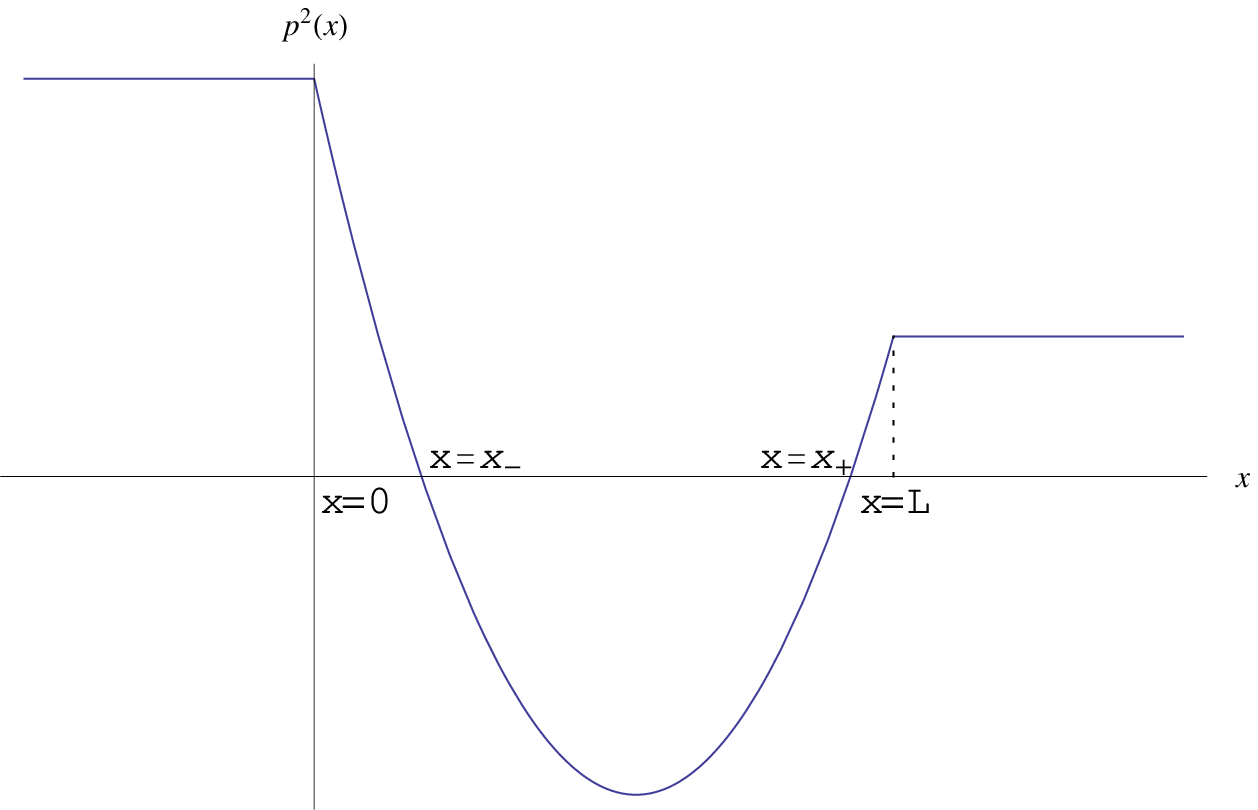}\label{fig:P:a}}
\subfigure[{~The plot of $p^{2}\left(  a\right)  $ for quantum cosmology. }] {\includegraphics[width=0.45\textwidth]{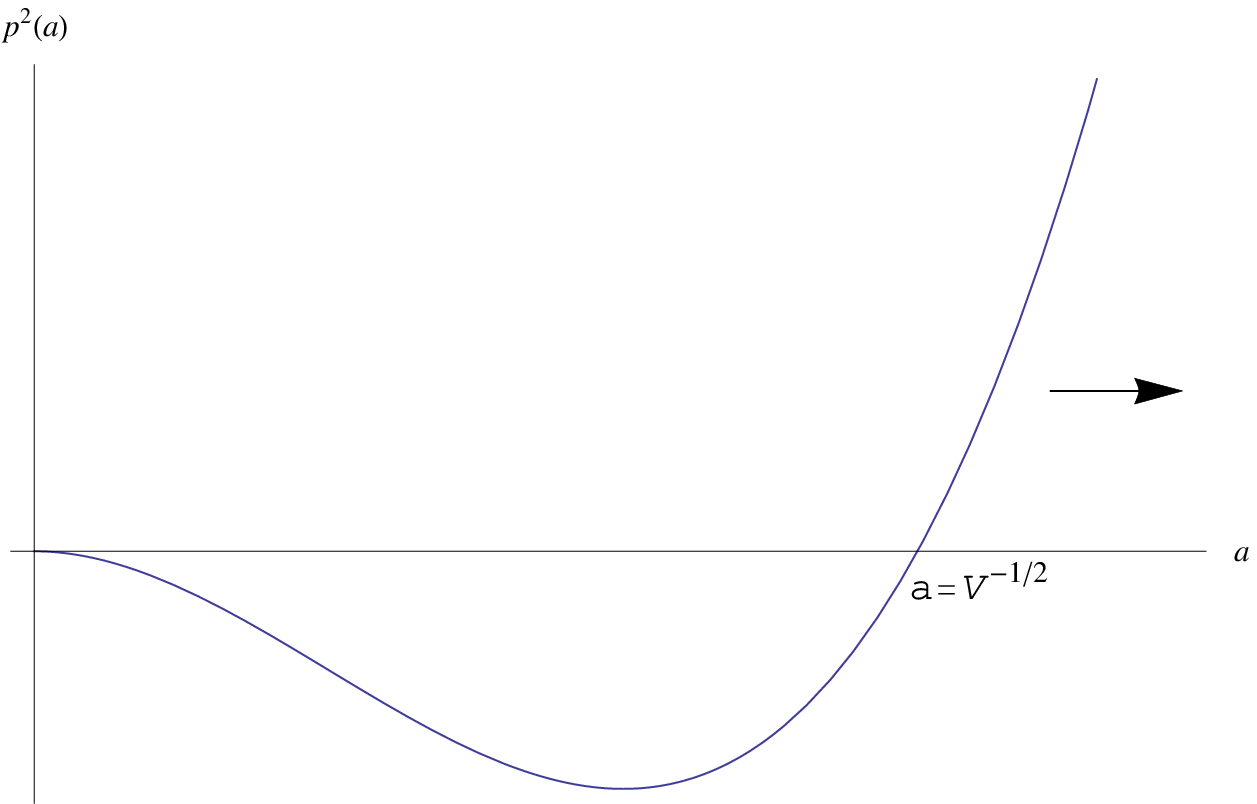}\label{fig:P:b}}
\end{center}
\caption{The plot of $p^{2}\left(  x\right)  $ and $p^{2}\left(  a\right)  $.}%
\label{fig:P}%
\end{figure}

For simplicity, we assume that electrons are described by a wave function
$\Psi\left(  t,x\right)  $ satisfying the $1+1$ dimensional deformed
Klein-Gordon equation
\begin{equation}
\left\{  \left[  i\partial_{t}+eV\left(  x\right)  \right]  ^{2}-\left(
i\partial_{x}\right)  ^{2}g^{2}\left(  -i\alpha\partial_{x}\right)
-m^{2}\right\}  \Psi\left(  t,x\right)  =0. \label{eq:KG}%
\end{equation}
Since the potential only depends on $x$, we could the following ansatz for
$\Psi\left(  t,x\right)  $%
\begin{equation}
\Psi\left(  t,x\right)  =e^{-iEt}\psi\left(  x\right)  .
\end{equation}
Substituting this expression in eqn. $\left(  \ref{eq:KG}\right)  $ results in
a deformed Schrodinger-like equation%
\begin{equation}
-\partial_{x}^{2}g^{2}\left(  -i\alpha\partial_{x}\right)  \psi\left(
x\right)  -p^{2}\left(  \xi\right)  \psi\left(  x\right)  =0,
\end{equation}
where $p^{2}\left(  x\right)  =\left[  E+eV\left(  x\right)  \right]
^{2}-m^{2}$. The function $p^{2}\left(  x\right)  $ has two turning points:%
\begin{equation}
x_{\pm}=\frac{1}{e\mathcal{E}}\left(  E\pm m\right)  ,
\end{equation}
where $0<x_{\_}<x_{+}<L$ since $m\leq E\leq e\mathcal{E}L-m$. We plot
$p^{2}\left(  x\right)  $ in FIG. \ref{fig:P:a}. Then, the WKB transmission
coefficient is given by%
\begin{equation}
T=\exp\left[  -2\int_{x_{\_}}^{x_{+}}\left\vert p\left(  x\right)  \right\vert
\lambda_{0}\left(  \alpha\left\vert p\left(  x\right)  \right\vert \right)
dx\right]  .
\end{equation}
The number of pairs produced per unit time with energies between $E$ and
$E+dE$ is%
\begin{equation}
\frac{dN}{dt}=2T\frac{dE}{2\pi}, \label{eq:dn/dt}%
\end{equation}
where the factor of $2$ takes into account the two polarizations of the
electron. Note that the turning points $x_{\pm}$ are the positions at which
the two particles of the pair are produced. Therefore, shifting the energy by
$dE$ results in a change in the positions of the particles by $dx=\frac
{dE}{e\mathcal{E}}$. It follows from eqn. $\left(  \ref{eq:dn/dt}\right)  $
that the pair production rate per unit length is%
\begin{equation}
W=\frac{e\mathcal{E}}{\pi}T\text{.}%
\end{equation}

Considering the $g\left(  x\right)  =\tan\left(  x\right)  /x$ case in which
$\lambda_{0}\left(  a\right)  =$arctanh$\left(  a\right)  /a$, we find%
\begin{equation}
T=\exp\left[  -\frac{2\pi}{e\mathcal{E}\alpha^{2}}\left(  1-\sqrt
{1-m^{2}\alpha^{2}}\right)  \right]  .
\end{equation}
Thus, the pair production rate per unit length is%
\begin{equation}
W=\frac{e\mathcal{E}}{\pi}\exp\left[  -\frac{2\pi}{e\mathcal{E}\alpha^{2}%
}\left(  1-\sqrt{1-m^{2}\alpha^{2}}\right)  \right]  .
\end{equation}
The Schwinger mechanism can explain the Unruh effect which predicts that an
accelerating observer will observe a thermal spectrum of photons and
particle--antiparticle pairs at temperature $T=\frac{a}{2\pi},$ where $a$ is
the acceleration \cite{EX-Unruh:1976db}. In fact, considering a free particle
of charge $e$ and mass $m$ moving in a static electric field $\mathcal{E}$,
one finds that the acceleration of the particle is $a=\frac{em}{\mathcal{E}}$.
It follows that the pair production rate per unit length is
\begin{equation}
W\sim\left[  -\frac{K}{a/2\pi}\frac{1-\sqrt{1-4K^{2}\alpha^{2}}}{2K^{2}%
\alpha^{2}}\right]  ,
\end{equation}
where we identify the reduced mass $\frac{m}{2}=K$ as the energy associated
with the pair production process. The Unruh temperature reads%
\begin{equation}
T_{u}\sim\frac{a}{2\pi}\frac{2K^{2}\alpha^{2}}{1-\sqrt{1-4K^{2}\alpha^{2}}}.
\label{eq:unruhT}%
\end{equation}

For a Schwarzschild black hole of the mass $M$, the event horizon is at
$r_{h}=2M$. Since the gravitational acceleration at the event horizon is given
by%
\begin{equation}
a=\frac{M}{r_{h}^{2}}=\frac{1}{4M},
\end{equation}
it follows from eqn. $\left(  \ref{eq:unruhT}\right)  $ that the Hawking
temperature is%
\begin{equation}
T_{h}\sim T_{0}\frac{2T_{h}^{2}\alpha^{2}}{1-\sqrt{1-4T_{h}^{2}\alpha^{2}}},
\label{eq:HawkingTemp}%
\end{equation}
where $T_{0}=\frac{1}{8\pi M}$, and we estimate that the energy of radiated
particles $K\sim T_{h}$. Solving the above equation for $T_{h}$ gives%
\begin{equation}
T_{h}=\frac{T_{0}}{1+\alpha^{2}T_{0}^{2}},
\end{equation}
which shows that in the $g\left(  x\right)  =\tan\left(  x\right)  /x$ case,
the quantum gravity effects always lower the Hawking temperature. Using the
first law of the black hole thermodynamics, we find that the black hole's
entropy is%
\begin{equation}
S=\int\frac{dM}{T_{h}}\sim\frac{A}{4}+\frac{\alpha^{2}}{8\pi}\ln\left(
\frac{A}{16\pi}\right)  , \label{eq:entropy}%
\end{equation}
where $A=4\pi r_{h}^{2}=16\pi M^{2}$ is the area of the horizon.

\subsection{P\"{o}schl-Teller Potential and Quasinormal Modes of A Black Hole}

It has been long known that the usual Schrodinger equation for the
P\"{o}schl--Teller potential of the form%
\begin{equation}
V_{PT}\left(  x,b\right)  =-\frac{V_{0}}{\cosh^{2}bx},
\end{equation}
is exactly solvable. For a particle of the mass $m=1/2$, the exact bound
states are given by
\begin{equation}
E_{n}\left(  b\right)  =-b^{2}\left[  -\left(  n+\frac{1}{2}\right)  +\left(
\frac{1}{4}+\frac{V_{0}}{b^{2}}\right)  ^{\frac{1}{2}}\right]  ^{2},
\label{eq:PTExact}%
\end{equation}
for $n=0,1,2,\cdots,N-1$, where $N+\frac{1}{2}>\left(  \frac{1}{4}+\frac
{V_{0}}{b^{2}}\right)  ^{\frac{1}{2}}$.\ We now use the WKB method to solve
the deformed Schrodinger equation for the bound states in the
P\"{o}schl--Teller potential. The deformed Schrodinger equation is given by%
\begin{equation}
-\partial_{x}^{2}g^{2}\left(  -i\alpha\partial_{x}\right)  \Psi\left(
x,b\right)  +\left[  V_{PT}\left(  x,b\right)  -E\left(  b\right)  \right]
\Psi\left(  x,b\right)  =0.
\end{equation}
The Bohr-Sommerfeld quantization condition $\left(  \ref{eq:BS}\right)  $ then
leads to%
\begin{equation}
\int_{-x_{0}}^{x_{0}}\left\vert p\left(  x\right)  \lambda_{1}\left(
\alpha\left\vert p\left(  x\right)  \right\vert \right)  \right\vert
dx\approx\left(  n+\frac{1}{2}\right)  \pi, \label{eq:PTBS}%
\end{equation}
where $x=\pm x_{0}\equiv\pm b^{-1}$arccosh$\left(  \frac{V_{0}^{1/2}}%
{\sqrt{-E}}\right)  $ are the turning points, and $p\left(  x\right)
=\sqrt{E-V_{PT}\left(  x,b\right)  }$.

Now consider the $g_{\pm}\left(  x\right)  =1\pm x^{2}$ case , in which we
find that the solutions of $sg_{\pm}\left(  -ias\right)  =i$ are
\begin{equation}
\lambda_{1}^{\pm}\left(  a\right)  =i\left(  1\pm a^{2}\right)  +\mathcal{O}%
\left(  a^{4}\right)  .
\end{equation}
Solving eqn. $\left(  \ref{eq:PTBS}\right)  $ for $E$, we find that the bound
states are%
\begin{equation}
E_{n}^{\pm}\left(  b\right)  \approx-b^{2}\left[  -\left(  n+\frac{1}%
{2}\right)  +\left(  \frac{V_{0}}{b^{2}}\right)  ^{\frac{1}{2}}\mp\alpha
^{2}b^{2}\left(  n+\frac{1}{2}\right)  ^{2}\left(  n+\frac{1}{2}-\frac{3}%
{2}\left(  \frac{V_{0}}{b^{2}}\right)  ^{\frac{1}{2}}\right)  \right]  ^{2},
\label{eq:PTWKB}%
\end{equation}
where $n=0,1,2,\cdots$ such that the sum of terms in the square bracket is
non-negative. If $\alpha=0$, from comparing the exact result $\left(
\ref{eq:PTExact}\right)  $ with the WKB one $\left(  \ref{eq:PTWKB}\right)  $,
it follows that the higher WKB corrections are suppressed by powers of
$b^{2}V_{0}^{-1}$. Therefore, combining results from eqns. $\left(
\ref{eq:PTExact}\right)  $ and $\left(  \ref{eq:PTWKB}\right)  $, one obtains
the bound states:%
\begin{align}
E_{n}^{\pm}\left(  b\right)   &  =-b^{2}\left[  -\left(  n+\frac{1}{2}\right)
+\left(  \frac{1}{4}+\frac{V_{0}}{b^{2}}\right)  ^{\frac{1}{2}}\right.
\nonumber\\
&  \left.  \mp\alpha^{2}b^{2}\left(  n+\frac{1}{2}\right)  ^{2}\left(
n+\frac{1}{2}-\frac{3}{2}\left(  \frac{V_{0}}{b^{2}}\right)  ^{\frac{1}{2}%
}+\mathcal{O}\left(  bV_{0}^{-1/2}\right)  \right)  +\mathcal{O}\left(
\alpha^{4}b^{4}\right)  \right]  ^{2}. \label{eq:PTeigen}%
\end{align}

To study quasinormal modes of a static and spherically symmetric, we consider
the propagation of the massless and minimally coupled scalar particles in a
general Schwarzschild-like metric:%
\begin{equation}
ds^{2}=h\left(  r\right)  dr^{2}-\frac{dr^{2}}{h\left(  r\right)  }%
-r^{2}d\Omega,
\end{equation}
where $d\Omega$ is the solid angle. The wave equation for the scalar particles
is the Klein-Gordon equation. After the wave function $\Psi\left(
t,r,\Omega\right)  $ is decomposed into eigenmodes of normal frequency
$\omega$ and angular momentum $l$, $\Psi\left(  t,r,\Omega\right)
=e^{-i\omega t}Y_{l}^{m}\left(  \Omega\right)  R\left(  r\right)  /r$, the
Klein-Gordon equation gives a Schrodinger-like equation for $R\left(
r\right)  $ in stationary backgrounds:%
\begin{equation}
-\frac{d^{2}R\left(  r_{\ast}\right)  }{dr_{\ast}^{2}}+V\left(  r_{\ast
}\right)  R=\omega^{2}R\left(  r_{\ast}\right)  , \label{eq:scheq}%
\end{equation}
where $dr_{\ast}=dr/h\left(  r\right)  $ is the tortoise coordinate, and
\begin{equation}
V\left(  r_{\ast}\right)  =h\left(  r\right)  \left(  \frac{l\left(
l+1\right)  }{r^{2}}+\frac{1}{r}\frac{dh}{dr}\right)  . \label{eq:V}%
\end{equation}
For asymptotically flat black holes, quasinormal modes are solutions of the
wave equation $\left(  \ref{eq:scheq}\right)  $, satisfying specific boundary
conditions \cite{EX-Konoplya:2011qq}%
\begin{equation}
R\sim e^{\pm i\omega r_{\ast}}\text{, }r_{\ast}\rightarrow\pm\infty.
\end{equation}

In deformed quantum mechanics, the Schrodinger-like equation $\left(
\ref{eq:scheq}\right)  $ could be changed to
\begin{equation}
\left[  -\frac{d^{2}}{dr_{\ast}^{2}}g^{2}\left(  \alpha\frac{d}{idr_{\ast}%
}\right)  +V\left(  r_{\ast}\right)  -\omega^{2}\right]  R\left(  r_{\ast
}\right)  =0. \label{eq:deformedScheq}%
\end{equation}
The quasinormal modes can be estimated by using a simpler potential
$-V_{PT}\left(  x,b\right)  $ that approximates $\left(  \ref{eq:V}\right)  $
closely, especially near its maximum \cite{EX-Ferrari:1984zz}. The quantities
$V_{0}$ and $b$ are given by the height and curvature of the potential
$V\left(  r\right)  $ at its maximum $r_{\ast}=r_{\ast,0}$:%
\begin{equation}
V_{0}=V\left(  r_{\ast,0}\right)  \text{ and }b^{2}=-\frac{1}{2V_{0}}%
\frac{d^{2}V\left(  r_{\ast}\right)  }{dr_{\ast}^{2}}|_{r_{\ast}=r_{\ast,0}%
}\text{.}%
\end{equation}
For a Schwarzschild black hole with $h\left(  r\right)  =1-\frac{2M}{r}$, we
find%
\begin{align}
V_{0}  &  =\frac{4l^{3}\left(  l+1\right)  ^{3}\left(  a_{l}-3-l-l^{2}\right)
\left[  1+3l\left(  l+1\right)  +a_{l}\right]  }{\left[  3l\left(  l+1\right)
-3+a_{l}\right]  ^{4}M^{2}},\label{eq:Vandb}\\
b^{2}  &  =\frac{16l^{2}\left(  l+1\right)  ^{2}\left\{  9\left(
-3+a_{l}\right)  +l\left(  l+1\right)  \left[  -33+4a_{l}+l\left(  l+1\right)
\left(  -13+9l\left(  l+1\right)  +3a_{l}\right)  \right]  \right\}  }{\left[
-3+3l\left(  l+1\right)  +{}a_{l}\right]  ^{4}\left[  1+3l\left(  l+1\right)
+{}a_{l}\right]  M^{2}},
\end{align}
where
\begin{equation}
a_{l}=\sqrt{9+l\left(  l+1\right)  \left(  14+9l\left(  l+1\right)  \right)
}\text{.}%
\end{equation}
Note that the ratio $bV_{0}^{-1/2}$ controlling the WKB expansion is given by%
\begin{equation}
bV_{0}^{-1/2}\sim l^{-1}\text{.}%
\end{equation}
In this approximation, eqn. $\left(  \ref{eq:deformedScheq}\right)  $ becomes%
\begin{equation}
\left[  -\frac{d^{2}}{dr_{\ast}^{2}}g^{2}\left(  \alpha\frac{d}{idr_{\ast}%
}\right)  -V_{PT}\left(  r_{\ast},b\right)  -\omega^{2}\left(  b\right)
\right]  R\left(  r_{\ast},b\right)  =0.
\end{equation}

To relate the quasinormal modes of the above equation to the bound states of
the P\"{o}schl--Teller potential, we consider the formal transformations
\cite{EX-Ferrari:1984zz}%
\begin{equation}
x\rightarrow-ir_{\ast}\text{ and }b\rightarrow ib
\end{equation}
such that $V_{PT}\left(  x,b\right)  =V_{PT}\left(  -ir_{\ast},ib\right)  $.
Let us define%
\begin{align}
\Psi\left(  x,b\right)   &  =R\left(  -ix,ib\right)  ,\nonumber\\
E\left(  b\right)   &  =-\omega^{2}\left(  ib\right)  .
\end{align}
Then $\Psi\left(  x,b\right)  $ satisfies
\begin{equation}
-\partial_{x}^{2}\tilde{g}^{2}\left(  -i\alpha\partial_{x}\right)  \Psi\left(
x,b\right)  +\left[  V_{PT}\left(  x,b\right)  -E\left(  b\right)  \right]
\Psi\left(  x,b\right)  =0,
\end{equation}
where $\tilde{g}\left(  x\right)  =g\left(  -ix\right)  $, and the boundary
conditions for the quasinormal modes are reduced to%
\begin{equation}
\Psi\left(  x,b\right)  \sim\exp\left(  \mp\sqrt{-E\left(  b\right)
}x\right)  \text{, as }x\rightarrow\pm\infty.
\end{equation}
The quasinormal modes in the $g\left(  x\right)  $ case can be found by the
bound states of the P\"{o}schl--Teller potential in the $\tilde{g}\left(
x\right)  $ case
\begin{equation}
\omega^{2}\left(  b\right)  =-E\left(  -ib\right)  .
\end{equation}
For the $g_{\pm}\left(  x\right)  =1\pm x^{2}$ case, it follows from eqn.
$\left(  \ref{eq:PTeigen}\right)  $ the quasinormal modes $\omega\equiv
\omega_{R}+i\omega_{I}$ of a Schwarzschild black hole in the deformed quantum
mechanics can be estimated as
\begin{align}
\left\vert \omega_{R}\right\vert  &  =\sqrt{V_{0}-\frac{b^{2}}{4}}\left\{
1\pm\frac{3\alpha^{2}b^{2}}{2}\left[  \left(  n+\frac{1}{2}\right)
^{2}+\mathcal{O}\left(  l^{-1}\right)  \right]  +\mathcal{O}\left(  \alpha
^{4}b^{4}\right)  \right\}  ,\nonumber\\
\omega_{I}  &  =-b\left(  n+\frac{1}{2}\right)  \left\{  1\pm\alpha^{2}%
b^{2}\left[  \left(  n+\frac{1}{2}\right)  ^{2}+\mathcal{O}\left(
l^{-1}\right)  \right]  +\mathcal{O}\left(  \alpha^{4}b^{4}\right)  \right\}
, \label{eq:QNM}%
\end{align}
where $n=0,1,2,\cdots$ such that the sum of terms in the square bracket of
eqn. $\left(  \ref{eq:PTeigen}\right)  $ is non-negative. If $l\gg1$, it
follows that eqn. $\left(  \ref{eq:QNM}\right)  $ work for $n<l$ when a
corresponding bound state exists. Our WKB method gives quite accurate results
for the regime of high multipole numbers $l$ of a Schwarzschild black hole of
the mass $M\gg1$, since $\alpha b\sim\alpha M^{-1}$. The P\"{o}schl-Teller
approximate potential method gives best result for low overtone number.
However for the higher modes, it is known that the P\"{o}schl-Teller potential
method gives higher values of $\omega_{I}$ \cite{EX-Nollert:1999ji}. In fact,
for $l\gg1\,$, eqn. $\left(  \ref{eq:Vandb}\right)  $ gives that
$b\approx\frac{1}{3\sqrt{3}M}$ in the P\"{o}schl-Teller approximate potential
method. On the other hand, the asymptotic quasinormal mode of a Schwarzschild
black hole is given by%
\begin{equation}
\omega\approx\frac{\ln3}{8\pi M}-\frac{i}{4M}\left(  n+\frac{1}{2}\right)  .
\end{equation}
It appears that if $b=1/4M$, we could have better approximations for
$\omega_{I}$ for the higher modes.

In \cite{EX-Hod:1998vk}, Hod used Bohr's correspondence principle to argue
that the highly damped black-hole oscillations frequencies were transitions
from an unexcited black hole to a black hole in a mode with $n\gg1$. Later,
Maggiore argued that these highly damped black-hole oscillations frequencies
should be interpreted as $\sqrt{\omega_{R}^{2}+\omega_{I}^{2}}$
\cite{EX-Maggiore:2007nq}. In high damping limit $n\gg1$, it is easy to see
that $\left\vert \omega_{R}\right\vert \ll\left\vert \omega_{I}\right\vert $,
and hence $\sqrt{\omega_{R}^{2}+\omega_{I}^{2}}\sim\left\vert \omega
_{I}\right\vert $. First consider the $\alpha=0$ case. It concludes from the
above arguments that the energy absorbed in the $n\rightarrow n-1$\ transition
with $n\gg1$\ is the minimum quantum that can be absorbed by the black hole.
Therefore, one obtains for the minimum quantum that%
\begin{equation}
\Delta M=\left\vert \omega_{I}\right\vert _{n}-\left\vert \omega
_{I}\right\vert _{n-1}=\frac{1}{4M},
\end{equation}
where we use $b=1/4M$. Since for a Schwarzschild black hole the horizon area
$A$ is related to the mass $M$ by $A=16\pi M^{2}$, a change $\Delta M$ in the
black hole mass produces a change
\begin{equation}
\Delta A=32\pi M\Delta M=8\pi,
\end{equation}
which coincides with the Bekenstein result \cite{EX-Bekenstein:1974jk}.

For the $g_{\pm}\left(  x\right)  =1\pm x^{2}$ case, it follows from eqn.
$\left(  \ref{eq:QNM}\right)  $ that for $n\gg1$, the minimum quantum absorbed
by the black hole is%
\begin{equation}
\Delta M=\left\vert \omega_{I}\right\vert _{n}-\left\vert \omega
_{I}\right\vert _{n-1}\approx\frac{1}{4M}\left[  1\pm\frac{3\alpha^{2}%
}{16M^{2}}\left(  n+\frac{1}{2}\right)  ^{2}\right]  ,
\end{equation}
which becomes negative or infinity as $n\rightarrow\infty$, depending on the
sign in front of $\alpha^{2}$. This means that contributions from higher order
terms become important and have to be included for very large value of $n$.
Despite the ignorance of higher order contributions, one may introduce an
upper cutoff $n_{c}$ on $n$, when higher order contributions are important.
Thus, the minimum quantum can be estimated as%
\begin{equation}
\Delta M\sim\frac{1}{4M}\left(  1\pm\frac{3\alpha^{2}n_{c}^{2}}{16M^{2}%
}\right)  ,
\end{equation}
which gives that the area of the horizon is quantized in units%
\begin{equation}
\Delta A=8\pi\left(  1\pm\frac{3\pi\alpha^{2}n_{c}^{2}}{A}\right)  .
\end{equation}
Since the minimum increase of entropy is $\ln2$ which is independent of the
value of the area, one then concludes that%
\begin{equation}
\frac{dS}{dA}\approx\frac{\Delta S}{\Delta A}\approx\frac{1}{4}\left(
1\mp\frac{3\pi\alpha^{2}n_{c}^{2}}{A}\right)  ,
\end{equation}
where a \textquotedblleft calibration factor\textquotedblright\ $\ln2/2\pi$ is
introduced in $\Delta A$ \cite{EX-Chen:2002tu}. From this, it follows that%
\begin{equation}
S\approx\frac{A}{4}\mp3\pi\alpha^{2}n_{c}^{2}\ln A,
\end{equation}
where the logarithmic term is the well known correction from quantum gravity
to the classical Bekenstein-Hawking entropy.

\subsection{Quantum Cosmology}

We now consider the case of a closed Friedmann universe with a scalar field
with the potential $V\left(  \Phi\right)  $. The Einstein-Hilbert action plus
the Gibbons--Hawking--York boundary term is
\begin{equation}
S_{g}=\frac{1}{4\kappa^{2}}\int_{M}d^{4}x\sqrt{-g}R+\frac{\varepsilon}%
{2\kappa^{2}}\int_{\partial M}d^{3}x\sqrt{\left\vert h\right\vert }K,
\end{equation}
where where $\kappa^{2}=4\pi$, $h_{ab}$ is the induced metric on the boundary,
$K$ is the trace of the second fundamental form, $\varepsilon$ is equal to $1$
when $\partial M$ is timelike, and $\varepsilon$ is equal to $-1$ when
$\partial M$ is spacelike. The action for the single scalar field is%
\begin{equation}
S_{m}=\int_{M}d^{4}x\sqrt{-g}\left(  -\frac{1}{2}g^{\mu\nu}\partial_{\mu}%
\Phi\partial_{\nu}\Phi-V\left(  \Phi\right)  \right)  \text{.}%
\end{equation}
The ansatz for the classical line element is%
\begin{equation}
ds^{2}=-N^{2}dt^{2}+a^{2}\left(  t\right)  d\Omega_{3}^{2},
\end{equation}
where $N\left(  t\right)  $ is the lapse function, and%
\begin{equation}
d\Omega_{3}^{2}=d^{2}\chi+\sin^{2}\chi\left(  d^{2}\theta+\sin^{2}\theta
d\phi^{2}\right)
\end{equation}
is the standard line element on $S^{3}$. Thus, one has for the curvature
scalar:%
\begin{equation}
R=\frac{6}{N^{2}}\left(  -\frac{\dot{N}\dot{a}}{Na}+\frac{\ddot{a}}{a}+\left[
\frac{\dot{a}}{a}\right]  ^{2}\right)  +\frac{6}{a^{2}}. \label{eq:R}%
\end{equation}
After partial integration of the second term in the parentheses of eqn.
$\left(  \ref{eq:R}\right)  $, we find that the minisuperspace action becomes%
\begin{equation}
S=S_{g}+S_{m}=\frac{1}{2}\int dtN\left(  -\frac{a\dot{a}^{2}}{N^{2}}+a\right)
+\int dtNa^{3}\left(  \frac{\dot{\Phi}^{2}}{2N^{2}}-\frac{V\left(
\Phi\right)  }{2}\right)  ,
\end{equation}
where we make rescalings
\begin{equation}
a\rightarrow\frac{\kappa}{\sqrt{6}\pi}a\text{, }N\rightarrow\frac{\kappa
}{\sqrt{6}\pi}N\text{, }\Phi\rightarrow\frac{\sqrt{3}}{\kappa}\Phi\text{, and
}V\left(  \Phi\right)  \rightarrow\frac{9\pi^{2}}{\kappa^{4}}V\left(
\Phi\right)  \text{.}%
\end{equation}
With the choice of the gauge $N=a$, the Hamiltonian can be written as%
\begin{equation}
\mathcal{H}=-\frac{\pi_{a}^{2}}{2}+\frac{\pi_{\Phi}^{2}}{2a^{2}}%
+\frac{\mathcal{V}\left(  a,\Phi\right)  }{2},
\end{equation}
where%
\begin{equation}
\mathcal{V}\left(  a,\Phi\right)  =a^{2}\left[  a^{2}V\left(  \Phi\right)
-1\right]  .
\end{equation}

To quantize this model, we could make the following replacements for $\pi_{a}$
and $\pi_{\Phi}$%
\begin{equation}
\pi_{a}^{2}=-\partial_{a}^{2}g^{2}\left(  -i\alpha\partial_{a}\right)  \text{
and }\pi_{\Phi}^{2}=-\partial_{\Phi}^{2}g^{2}\left(  -i\alpha\partial_{\Phi
}\right)  \text{.}%
\end{equation}
Then, the Wheeler--DeWitt equation, $\mathcal{H}\psi\left(  a,\Phi\right)
=0$, reads%
\begin{equation}
\left[  \partial_{a}^{2}g^{2}\left(  -i\alpha\partial_{a}\right)
-\frac{\partial_{\Phi}^{2}g^{2}\left(  -i\alpha\partial_{\Phi}\right)  }%
{a^{2}}+\mathcal{V}\left(  a,\Phi\right)  \right]  \psi\left(  a,\Phi\right)
=0\text{.} \label{eq:WD}%
\end{equation}
Confining ourselves to regions in which the potential can be approximated by a
cosmological constant, we can drop the term involving derivatives with respect
to $\Phi$ in eqn. $\left(  \ref{eq:WD}\right)  $, thereby obtaining a simple
$1$-dimensional problem which is amenable to the WKB analysis. In this case,
eqn. $\left(  \ref{eq:WD}\right)  $ becomes eqn. $\left(  \ref{eq:DeformedEq}%
\right)  $ with $p^{2}\left(  a\right)  =\mathcal{V}\left(  a,\Phi\right)  $.
The function $p^{2}\left(  a\right)  $ is illustrated in FIG. \ref{fig:P:b}.
Thus, we find that the WKB solutions are%
\begin{equation}
\psi_{WKB}\left(  a,\Phi\right)  =C_{1}\left(  \Phi\right)  \psi_{WKB}%
^{1}\left(  a,\Phi\right)  +C_{3}\left(  \Phi\right)  \psi_{WKB}^{3}\left(
a,\Phi\right)  \text{ for }a^{2}V\left(  \Phi\right)  >1
\end{equation}
and
\begin{equation}
\psi_{WKB}\left(  a,\Phi\right)  =C_{0}\left(  \Phi\right)  \psi_{WKB}%
^{0}\left(  a,\Phi\right)  +C_{2}\left(  \Phi\right)  \psi_{WKB}^{2}\left(
a,\Phi\right)  \text{ for }a^{2}V\left(  \Phi\right)  <1,
\end{equation}
where%
\begin{equation}
\psi_{WKB}^{k}\left(  \xi\right)  =\frac{\exp\left[  \int_{V^{-1}\left(
\Phi\right)  }^{a}a^{\prime}\sqrt{\left\vert a^{\prime2}V\left(  \Phi\right)
-1\right\vert }\lambda_{k}\left(  \alpha a^{\prime}\sqrt{\left\vert
a^{\prime2}V\left(  \Phi\right)  -1\right\vert }\right)  da^{\prime}\right]
}{\sqrt{\left\vert \left[  x^{2}g^{2}\left(  \alpha x\right)  \right]
^{\prime}|_{x=-ia\sqrt{\left\vert a^{2}V\left(  \Phi\right)  -1\right\vert
}\lambda_{k}\left(  \alpha a\sqrt{\left\vert a^{2}V\left(  \Phi\right)
-1\right\vert }\right)  }\right\vert }}.
\end{equation}
To specify the WKB solution of the Wheeler--DeWitt equation, we need to make a
choice of boundary condition. The tunneling proposal was proposed by Vilenkin
\cite{EX-Vilenkin:1984wp}, which states that the universe tunnels into
\textquotedblleft existence from nothing.\textquotedblright\ The tunneling
proposal of Vilenkin \cite{EX-Vilenkin:1987kf} is that the wavefunction $\psi$
should be everywhere bounded, and at singular boundaries of superspace $\psi$
includes only outgoing modes. In our case, the boundary $a=\infty$ with $\phi$
finite is the singular boundary. For the WKB solutions, it follows that the
tunneling proposal demands that only the outgoing modes with
\begin{equation}
\psi_{WKB}\left(  a,\Phi\right)  =C\left(  \Phi\right)  \psi_{WKB}^{1}\left(
a,\Phi\right)
\end{equation}
are admitted in the oscillatory region $a^{2}V\left(  \Phi\right)  >1$. Note
that since $\frac{\partial\mathcal{V}\left(  a,\Phi\right)  }{\partial
a}|_{a=V^{-1}\left(  \Phi\right)  }>0$, we have that $\psi_{WKB}^{1}\left(
a,\Phi\right)  $ propagates away from the turning point $a=V^{-1}\left(
\Phi\right)  $, and $\psi_{WKB}^{2}\left(  \xi\right)  $ is exponentially
increasing away from $a=V^{-1}\left(  \Phi\right)  $. The WKB connection
formula $\left(  \ref{eq:wkbC2}\right)  $ gives the wave function $\psi
_{WKB}\left(  a,\Phi\right)  $ in the classically forbidden region
$a^{2}V\left(  \Phi\right)  <1$:%
\begin{equation}
\psi_{WKB}\left(  a,\Phi\right)  =e^{3i\pi/4}C\left(  \Phi\right)  \psi
_{WKB}^{2}\left(  a,\Phi\right)  .
\end{equation}

It appears that eqn. $\left(  \ref{eq:WD}\right)  $ can be described as a
particle of zero energy moving in a potential $\mathcal{V}\left(
a,\Phi\right)  $. The universe can start at $a=0$ and tunnel through the
potential barrier to the oscillatory region. The tunneling probability are
given by eqn. $\left(  \ref{eq:tunneling}\right)  $:%
\begin{equation}
P\left(  \Phi\right)  \sim\exp\left[  -2\int_{0}^{V^{-1/2}\left(  \Phi\right)
}a\sqrt{1-a^{2}V\left(  \Phi\right)  }\lambda_{0}\left(  \alpha a\sqrt
{1-a^{2}V\left(  \Phi\right)  }\right)  da\right]  .
\end{equation}
$P\left(  \Phi\right)  $ can be interpreted as the probability distribution
for the initial values of $\Phi$ in the ensemble of nucleated universes. For a
chaotic potential $V\left(  \Phi\right)  =\lambda\Phi^{2p}$, there will then
be a minimum value of the scalar field, $\Phi_{s}$, for which sufficient
inflation is obtained. The probability of sufficient inflation is given by a
conditional probability:%
\begin{equation}
\mathcal{P}\left(  \Phi>\Phi_{s}|\Phi_{1}<\Phi<\Phi_{2}\right)  =\frac
{\int_{\Phi_{s}}^{\Phi_{2}}P\left(  \Phi\right)  d\Phi}{\int_{\Phi_{1}}%
^{\Phi_{2}}P\left(  \Phi\right)  d\Phi},
\end{equation}
where the initial value of $\Phi$ lies in the range $\Phi_{1}<\Phi<\Phi_{2}$,
and the values $\Phi_{1}$ and $\Phi_{2}$ are respectively lower and upper
cutoffs on the allowed values of $\Phi$.

For $g_{\pm}\left(  x\right)  =1\pm x^{2}$, we find%
\begin{equation}
P^{\pm}\left(  \Phi\right)  \sim\exp\left[  -\frac{2}{3V\left(  \Phi\right)
}\left(  1\pm\frac{6\alpha^{2}}{35V\left(  \Phi\right)  }+\mathcal{O}\left(
\alpha^{4}\right)  \right)  \right]  .
\end{equation}
The probability of sufficient inflation is%
\begin{equation}
\mathcal{P}\left(  \Phi>\Phi_{s}|\Phi_{1}<\Phi<\Phi_{2}\right)  =\frac
{\int_{\Phi_{s}}^{\Phi_{2}}\exp\left[  -\frac{2}{3V\left(  \Phi\right)
}\right]  d\Phi}{\int_{\Phi_{1}}^{\Phi_{2}}\exp\left[  -\frac{2}{3V\left(
\Phi\right)  }\right]  d\Phi}\left\{  1\pm\frac{6\alpha^{2}}{35}\left[
F\left(  \Phi_{1}\right)  -F\left(  \Phi_{s}\right)  \right]  \right\}  ,
\end{equation}
where we define%
\begin{equation}
F\left(  \phi\right)  =\frac{\int_{\phi}^{\Phi_{2}}\exp\left[  -\frac
{2}{3V\left(  \Phi\right)  }\right]  \frac{d\Phi}{V^{2}\left(  \Phi\right)  }%
}{\int_{\phi}^{\Phi_{2}}\exp\left[  -\frac{2}{3V\left(  \Phi\right)  }\right]
d\Phi}.
\end{equation}
Since $F^{\prime}\left(  \phi\right)  <0$, we find that $F\left(  \Phi
_{1}\right)  >F\left(  \Phi_{s}\right)  $, and the probability of sufficient
inflation is higher/lower in the $g_{+}\left(  x\right)  $/$g_{-}\left(
x\right)  $ case than in the usual case.

\section{Conclusion}

\label{Sec:Con}

In the first part of this paper, we used the WKB approximation method to
approximately solve the deformed Schrodinger-like differential equation
$\left(  \ref{eq:DeformedEq}\right)  $ and applied the steepest descent method
to find the exact solutions around turning points. Matching the two sets of
solutions in the overlap regions, we obtained the WKB connection formulas
through a turning point, the deformed Bohr--Sommerfeld quantization rule, and
tunneling rate formula.

In the second part, several examples of applying the WKB approximation to the
deformed quantum mechanics were discussed. In the example of the harmonic
oscillator, we used the WKB approximation to calculate bound states in the
$g\left(  x\right)  =\tan x/x$ case. After compared with the exact solutions,
our WKB ones were shown to agree with the leading term of the WKB expansion of
the exact result.

The pair production rate of electron--positron pairs by a strong electric
field was computed in the case with $g\left(  x\right)  =\tan x/x\approx
1+x^{2}/3$ and found to be%
\begin{equation}
W\sim\exp\left[  -\frac{\pi m^{2}}{e\mathcal{E}}\left(  1+\frac{\alpha
^{2}m^{2}}{4}\right)  \right]  .
\end{equation}
In the GUP case with $g\left(  x\right)  =1+x^{2}/3$, the scalar particles
pair creation rate by an electric field was calculated in the context of the
deformed QFT \cite{CON-Mu:2015pta} and given by%
\begin{equation}
W\sim\exp\left[  -\frac{\pi m^{2}}{e\mathcal{E}}\left(  1+\frac{4\alpha
^{2}e^{2}\mathcal{E}^{2}}{3\pi^{2}m^{2}}\right)  \right]  .
\end{equation}
The pair creation rate was also calculated by using Bogoliubov transformations
\cite{CON-Haouat:2013yba} and given by%
\begin{equation}
W\sim\exp\left[  -\frac{m^{2}\pi}{e\mathcal{E}}\left(  1+\frac{\alpha^{2}%
m^{2}}{4}\left(  1-\frac{e^{2}\mathcal{E}^{2}}{m^{4}}\right)  \right)
\right]  .
\end{equation}
Although the expressions for the quantum gravity correction are different in
the above cases, the effects of the minimal length all tend to lower the pair
creation rates.

Using the Bohr--Sommerfeld quantization rule, we calculated the bound states
of the P\"{o}schl-Teller potential in the $g\left(  x\right)  =1\pm x^{2}$
case. The quasinormal modes of a black hole could be related to the bound
states of the P\"{o}schl-Teller potential by approximating the gravitational
barrier potential of the black hole with the inverted P\"{o}schl-Teller
potential. In this way, the effects of quantum gravity on quasinormal modes of
a Schwarzschild black hole were estimated. Moreover, the effects of quantum
gravity on the area quantum of the black hole was considered via Bohr's
correspondence principle. In the $g\left(  x\right)  =1\pm x^{2}$ case, we
found that the minimum increase of area was%
\begin{equation}
\Delta A=8\pi\left(  1\pm\frac{3\pi\alpha^{2}n_{c}^{2}}{A}\right)  ,
\end{equation}
where $n_{c}$ is some upper cutoff on $n$. On the other hand, authors of
\cite{CON-AmelinoCamelia:2005ik} followed the original Bekenstein argument
\cite{CON-Bekenstein:1973ur} and gave that%
\begin{equation}
\Delta A=8\pi\left(  1\pm\frac{6\pi\alpha^{2}}{A}\right)
\end{equation}
in the MDR scenario in which $g\left(  x\right)  =1\pm x^{2}$, and
\begin{equation}
\Delta A=8\pi\left(  1+\frac{4\pi\alpha^{2}}{A}\right)
\end{equation}
in the GUP scenario in which $g\left(  x\right)  =1+x^{2}$.

Finally, we used the WKB approximation method to find the WKB solutions of the
deformed Wheeler--DeWitt equation for a closed Friedmann universe with a
scalar field. In the context of the tunneling proposal, the effects of quantum
gravity on the probability of sufficient inflation was also discussed.

\begin{acknowledgments}
We are grateful to Houwen Wu and Zheng Sun for useful discussions. This work
is supported in part by NSFC (Grant No. 11375121 11005016 and 11175039).
\end{acknowledgments}

\appendix

\section{Contours in $g(x)=\frac{\tan(x)}{x}$ Case}

\begin{figure}[tb]
\begin{centering}
\includegraphics[scale=0.3]{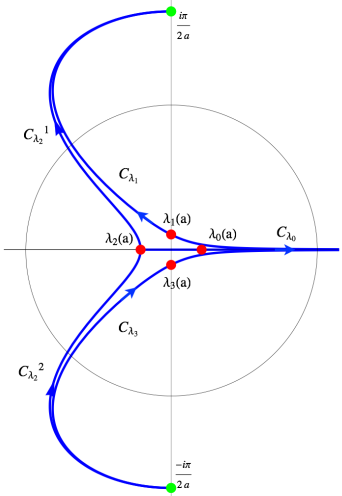}
\par\end{centering}
\caption{Contours (blue lines) and saddle points (red dots) of $\psi_{1}
(\rho)$ and $\psi_{2}(\rho)$ in the $g(x)=\frac{\tan(x)}{x}$ case. The green
dots are poles where the contours end.}%
\label{contourT}%
\end{figure}

In this appendix, we consider the contours $C_{>}$ and $C_{<}$ in the
$g(x)=\frac{\tan(x)}{x}$ case. In this case, we have four regular saddle
points:
\begin{equation}
\lambda_{k}(a)=\frac{\text{arctanh}{\left(  e^{i\pi k/2}a\right)  }}{a}\text{,
for }k=0,1,2,3.
\end{equation}
Note that $f_{\pm}(s)$ both have poles at $s=\pm\frac{i\pi}{2a}$. The
endpoints of a contour are at either infinity or singularities. It turns out
that in the $g(x)=\frac{\tan(x)}{x}$ case, the contours $C_{>}$ and $C_{<}$
both start from and terminate at poles $s=\pm\frac{i\pi}{2a}$. Following the
conventions adopted in section \ref{Sec:WKB}, we plot the contours $C_{>}$ and
$C_{<}$ in FIG. \ref{contourT}, where the green dots are poles. The circle
$C_{R}$ is also plotted in FIG. \ref{contourT}.

\end{document}